\documentclass[11pt,a4paper]{article}

\usepackage[margin=2.5cm]{geometry}

\usepackage{lmodern}
\usepackage{iftex}
\ifpdftex
\usepackage[utf8]{inputenc}
\usepackage[T1]{fontenc}
\usepackage{textcomp}
\fi

\usepackage{amsmath}        
\usepackage{reptheorem}     
\usepackage{mathtools}
\usepackage{amsfonts}       
\usepackage{amssymb}
\usepackage{amsthm}         
\usepackage{bm}             
\usepackage{booktabs}       
\usepackage{caption}        
\usepackage{dcolumn}        
\usepackage{floatrow}       
\usepackage{graphicx}       
\usepackage[nopatch=item]{microtype}   
\usepackage{paralist}       
\usepackage{enumitem}
\usepackage[dvipsnames]{xcolor}         
\usepackage{adjustbox} 
\usepackage{placeins} 
\usepackage{array} 
\usepackage{tabularx} 

\usepackage{makecell}

\usepackage{pifont}
\newcommand{\xmark}{\ding{55}}
\newcommand{\cmark}{\ding{51}}

\usepackage{hyperref}
\hypersetup{unicode}
\hypersetup{breaklinks=true}

\usepackage{algorithm}
\usepackage[noEnd=true,spaceRequire=false]{algpseudocodex}
\algrenewcommand\algorithmicthen{\hspace{-1mm}\textbf{:}}
\algrenewcommand\algorithmicelse{\textbf{else:}}
\algrenewcommand\algorithmicdo{\hspace{-1mm}\textbf{:}}
\usepackage{fancyvrb}       
\usepackage{thm-restate}	

\usepackage{cleveref}
\crefname{rule}{rule}{rules}
\crefname{inv}{invariant}{invariants}
\crefname{obs}{observation}{observations}

\makeatletter
\AtBeginDocument{\@ifundefined{Hy@theorem@refstepcounter}{}{%
  \@ifundefined{cref@thmnoarg}{}{%
    \patchcmd{\cref@thmnoarg}{\refstepcounter}{\Hy@theorem@refstepcounter}{}%
      {\PackageWarningNoLine{main}{Could not patch the theorem link targets}}}}}
\makeatother

\usepackage[natbib,style=iso-alphabetic]{biblatex} 
\DeclareFieldFormat[article]{citetitle}{\mkbibemph{#1}}
\DeclareFieldFormat[article]{title}{\mkbibemph{#1}}
\DeclareFieldFormat[article]{journaltitle}{\normalfont{#1}}

\usepackage{authblk}
\usepackage{orcidlink}

\theoremstyle{plain}

\AfterEndEnvironment{theorem}{\noindent\ignorespaces}
\newtheorem{lemma}{Lemma}
\AfterEndEnvironment{lemma}{\noindent\ignorespaces}

\AfterEndEnvironment{claim}{\noindent\ignorespaces}
\newtheorem{defn}{Definition}
\AfterEndEnvironment{defn}{\noindent\ignorespaces}
\newtheorem{cor}{Corollary}
\AfterEndEnvironment{cor}{\noindent\ignorespaces}
\newtheorem{obs}{Observation}
\AfterEndEnvironment{obs}{\noindent\ignorespaces}
\newtheorem{fact}{Fact}
\AfterEndEnvironment{fact}{\noindent\ignorespaces}

\AfterEndEnvironment{assm}{\noindent\ignorespaces}

\AfterEndEnvironment{nott}{\noindent\ignorespaces}

\theoremstyle{remark}
\newtheorem*{rem}{Remark}

\ifcsname DeclareCaptionStyle\endcsname
\DeclareCaptionStyle{thesis}{style=base,font=small,labelfont=bf,labelsep=quad}
\fi

\DefineVerbatimEnvironment{code}{Verbatim}{fontsize=\small, frame=single}

\ifcsname lstset\endcsname
\fi

\ifcsname DeclareNewFloatType\endcsname
\DeclareNewFloatType{listing}{}
\floatname{listing}{Program}
\fi

\DeclareMathOperator{\pr}{\normalfont{\textsf{P}}}

\def\O{\mathcal{O}}

\newcommand{\abs}[1]{\left|{#1}\right|}
\newcommand{\floor}[1]{\left\lfloor{#1}\right\rfloor}
\newcommand{\ceil}[1]{\left\lceil{#1}\right\rceil}

\newcommand{\eqdef}{\stackrel{\mathrm{def}}{=}}

\DeclareDocumentCommand{\nlog}{o o m}{\ln_{#1}^{#2}{\hspace{-0.2em}\left(#3\right)}}
\DeclareDocumentCommand{\xlog}{o o m}{\log_{#1}^{#2}{\hspace{-0.2em}\left(#3\right)}}

\DeclareMathOperator{\err}{\normalfont{Err}}
\DeclareMathOperator{\rank}{\normalfont{R}}
\newcommand*{\xhat}[1]{#1\kern-0.8em\hat{\phantom{#1}}}
\DeclareMathOperator{\estrank}{\normalfont{\xhat{R}}}

\DeclareMathOperator{\bin}{\normalfont{SUM}}

\DeclareMathOperator{\mar}{\mathcal{M}}

\DeclareMathOperator{\se}{\mathrm{sec}}

\ifx\citet\undefined\else

\DeclareDelimFormat[textcite]{multinamedelim}{\addcomma\space}
\DeclareDelimFormat[textcite]{finalnamedelim}{\space and~}

\fi

\title{Almost Optimal Relative Error Streaming Quantiles with Full Mergeability
	\thanks{Supported by the ERC CZ project LL2406 of the Czech Ministry of Education, Youth and Sports, 
	by grant 24-10306S of GA \v{C}R,
	and by Center for Foundations of Modern Computer Science (Charles Univ.\ project UNCE 24/SCI/008).}
}
\date{}
\author{Tomáš Domes\,\orcidlink{0009-0008-9765-1809} and Pavel Veselý\,\orcidlink{0000-0003-1169-7934}}
\affil{Computer Science Institute of Charles University, Czech Republic\\
	\texttt{\{domestomas,vesely\}@iuuk.mff.cuni.cz}
}

\begin{document}

\maketitle

\noindent 
Quantile summaries provide a scalable way to estimate the distribution of individual attributes in large datasets that are often distributed across multiple machines or generated by sensor networks.
The state-of-the-art summaries with \emph{relative error guarantees} of $\pm \varepsilon\cdot \mathrm{rank}$, offering increasingly higher accuracy towards the distribution’s tails, are ReqSketch~\cite{ReqSketch} and the recent streaming algorithm by Gribelyuk, Sawettamalya, Wu, and Yu~\cite{GribelyukSWY25}.
For inputs of size $N$ and probability $\delta$ of exceeding the error bound,
ReqSketch achieves space of $\O\left(\varepsilon^{-1}\cdot \log^{1.5}(\varepsilon N)\cdot \sqrt{\log \delta^{-1}}\right)$ memory words, and offers \emph{full mergeability}, allowing distributed or parallel processing of datasets.
In contrast, the latter algorithm exceeds the lower bound of $\Omega(\varepsilon^{-1}\cdot \log(\varepsilon N)\cdot \log\log\delta^{-1})$ from~\cite{add_lower} only by a factor of $\O\left(\log\varepsilon^{-1}\cdot (\log \varepsilon^{-1} + \log\log N)\cdot\log^2\log\delta^{-1}\right)$,
but it lacks the mergeability property, is substantially more involved than ReqSketch, and has unclear practical performance.
Both of these data structures consist of relative compactors from~\cite{ReqSketch}
that subsample the stream in a structured way and work well in the streaming setting but are difficult to analyze under mergeability.

We design a new relative-error quantile summary that achieves full mergeability and also improves the state-of-the-art streaming space bound to
$$
\O\left(\varepsilon^{-1}\cdot \log(\varepsilon N) \cdot\log \log \delta^{-1} \cdot (\log \varepsilon^{-1} + \log \log(\varepsilon N) + \log \log \delta^{-1})\right)\,,
$$ 
which is substantially closer to the lower bound.
At the same time, our data structure is simpler than the algorithm from~\cite{GribelyukSWY25} and therefore more suitable for implementation.

Our technical contribution is twofold: First, we design an \emph{adaptive compactor}, which improves its state when small items are inserted and is well suited for mergeability, compared to the relative compactor that behaves independently of the incoming items.
Second, we design a new \emph{space allocation strategy}, which assigns capacities to compactors and exploits their adaptivity.
While proving the error guarantees is relatively straightforward,
the main challenge is in analyzing space requirements, for which we design a new potential-function argument.

\section{Introduction}\label{intro}

Quantile summaries, or sketches, are data structures of sublinear size, typically just polylogarithmic in the input length,
allowing us to efficiently process large streaming or distributed datasets and
approximate their distributions by estimating ranks and quantiles.
Here, the rank of an item $y$ with respect to $\mathcal{S}$, denoted $\rank(y, \mathcal{S})$,
is the number of items in $\mathcal{S}$ that are smaller than or equal to $y$,
and quantiles are the inverse of ranks, i.e., a $\varphi$-quantile for $\varphi\in [0, 1]$ is the $\lceil \varphi\cdot |\mathcal{S}|\rceil$-th smallest input item.
Quantile summaries are essential for many applications in network monitoring, e.g.,
tracking latencies~\cite{MassonRL19,tene2015_latency_talk} or round-trip times~\cite{CormodeKMS05}, or in sensor networks~\cite{q_digest,HeCCW15}.

Many of these applications require two main properties from the summaries that are both challenging to achieve while preserving high space efficiency:
First, the quantile summary should provide an accurate distribution approximation, especially for the tails.
In particular, network latencies are typically heavily long-tailed~\cite{MassonRL19},
and understanding the tail requires increasingly more accurate estimates of the 95th, 99th, 99.5th, etc., percentiles.
The notion of \emph{relative error} captures this scenario by requiring that the rank error for an item of rank $R$ is at most $\varepsilon\cdot R$,
for an accuracy parameter $\varepsilon > 0$;
equivalently, for a quantile query $\varphi\in [0, 1]$, the summary should return a $\hat{\varphi}$-quantile, for  $\hat{\varphi} = (1 \pm \varepsilon)\cdot \varphi$.
While this yields lower error for items of small rank, high accuracy for the other tail, as desired in monitoring latencies, can be achieved by simply flipping the comparator.
A long line of work~\cite{GuptaZ03,CormodeKMS05,q_digest_new,ZhangLXKW06,zhang,ReqSketch} studied relative-error quantile summaries,
and recently Gribelyuk, Sawettamalya,
Wu, and Yu~\cite{GribelyukSWY25} developed a randomized streaming algorithm using near-optimal space of $\widetilde{O}(\varepsilon^{-1}\cdot \log \varepsilon N)$\footnote{
	The notation $\widetilde{O}(f)$ hides factors polylogarithmic in $f$, and the dependence on $\delta$, the probability of a too-large error for a single query.
}.

Second, processing the dataset in one pass is, in many cases, infeasible, 
as it may be distributed across many machines or consist of observations aggregated by sensors~\cite{q_digest}.
To allow for distributed or parallel processing of large datasets,
the sketch should come with a merge operation that takes two sketches representing datasets $\mathcal{S}_1$ and $\mathcal{S}_2$ and outputs one sketch representing the multiset union of $\mathcal{S}_1$ and $\mathcal{S}_2$.
A quintessential property of summaries is \emph{full mergeability}~\cite{agarwal2013mergeable},
stating that even when the sketch is built by \emph{any} sequence of pairwise merge operations on single-item sketches obtained from a set $\mathcal{S}$,
its error should be bounded, ideally achieving the same accuracy-space trade-off as in the streaming model.
The mergeability setting is more general, since the streaming model can be seen as a repeated merge with a single-item sketch.

The state-of-the-art relative-error and fully mergeable quantile summary is ReqSketch~\cite{ReqSketch},
which consists of a sequence of \emph{relative compactors}. These are buffers that randomly subsample the stream in a way that facilitates the relative error; see Section~\ref{technical_overview}.
ReqSketch has been implemented within the Apache DataSketches library~\cite{datasketches},
due to its overall simplicity, mergeability, and fast update time, compared to the widely used $t$-digest~\cite{tdigest,ReqSketch_improvements}.

While ReqSketch~\cite{ReqSketch} only achieves a suboptimal space bound of $\O(\varepsilon^{-1}\cdot \log^{1.5} \varepsilon N)$ memory words,
the recent improvement in~\cite{GribelyukSWY25}
does not come with any mergeability guarantees (or even a merge operation), and moreover, the new algorithm from~\cite{GribelyukSWY25} is also significantly more involved than ReqSketch.
Indeed, to the best of our knowledge, there is no implementation, and it seems likely that ReqSketch would still perform better than
the algorithm of~\cite{GribelyukSWY25} for data streams in practice.

\subsection{Our result}
The key obstacle to progress appears to be proving mergeability for relative-error sketches.
Indeed, despite the simplicity of the ReqSketch algorithm and its relatively accessible analysis in the streaming setting,
the proof of full mergeability in~\cite{ReqSketch}
is very complicated. Furthermore, the algorithm of~\cite{GribelyukSWY25} does not provide a merge operation at all.
This makes it difficult to either improve the space bound for fully mergeable sketches or generalize the proof for other error guarantees.

Here, we develop \emph{adaptive compactors} that naturally capture the charging required to bound the error under full mergeability.
When used instead of relative compactors in ReqSketch,
the new compactors allow for a significantly simplified proof of full mergeability,
while preserving its other properties like space requirements, update time, and relative-error guarantees{\interfootnotelinepenalty=10000 \footnote{
	See the previous version of our paper with the same space requirements as ReqSketch at \url{https://doi.org/10.48550/arXiv.2511.17396}.
}}.

Moreover, our compactors improve their state after the insertion of small items; in contrast, the state of relative compactors depends only on the length of the input stream. Taking advantage of this crucial property and 
using a new space allocation strategy inspired by~\cite{GribelyukSWY25},
we obtain an even better space bound than~\cite{GribelyukSWY25}, while keeping the data structure much simpler.
In particular, we have only one compactor per level of the sketch, as in ReqSketch~\cite{ReqSketch} or the uniform-error KLL sketch~\cite{KLL}. In contrast, the streaming algorithm from~\cite{GribelyukSWY25} needs a subsketch at each level of the final sketch, and this subsketch consists of multiple relative compactors.
Our sketch is fully mergeable, and the introduction of the merge operation adds very little complexity to both the algorithm and the proofs. The theoretical properties of the sketch are formally stated in our main theorem:

\repthm{thm:mainthm}

\Cref{thm:mainthm} gives guarantees for an arbitrary single query. To get the relative error bound for all queries $y\in \mathcal{U}$ simultaneously, i.e., the so-called \emph{all-quantiles guarantee}, we simply use the union bound over
a suitable $\varepsilon$-net, as described in Appendix~B of~\cite{ReqSketch}.
Namely, \Cref{thm:mainthm} with $\delta = \delta' \cdot \varepsilon / \log(\varepsilon N)$ gives $1 \pm \O(\varepsilon)$
multiplicative error for any rank query simultaneously with probability at least $1-\delta'$ using a sketch size of
$$
\O\left(\varepsilon^{-1} \cdot \log(\varepsilon N)\cdot  \log \log \left(\frac{\log(\varepsilon N)}{\varepsilon\cdot \delta'}\right) \cdot (\log \log \delta'^{-1} + \log \varepsilon^{-1} + \log \log(\varepsilon N))\right)\,.
$$

\subsection{Comparison to previous results and lower bounds}\label{lower_bounds}

\paragraph{Comparison to previous randomized algorithms}
Previously, the state-of-the-art space bound with full mergeability was
$\O(\varepsilon^{-1}\cdot \log^{1.5}(\varepsilon N)\cdot \sqrt{\log \delta^{-1}})$
by Cormode, Karnin, Liberty, Thaler, and Veselý~\cite{ReqSketch}, and in the streaming setting it was
$\O(\varepsilon^{-1}\cdot \log(\varepsilon N)\cdot  \log^3(\ln \delta^{-1})\cdot \log\varepsilon^{-1}\cdot (\log \log N + \log \varepsilon^{-1}))$
by Gribelyuk, Sawettamalya, Wu, and Yu~\cite{GribelyukSWY25};
these bounds are for an arbitrary single query. For a quick comparison, see \Cref{comparison_space}.

Here, we note that the error analysis in~\cite{ReqSketch} has a subtle (fixable) issue in analyzing the probability of a too-large error. We explain it in Section~\ref{sec:error_bound} (see the remark right before \Cref{time}) and describe how our analysis fixes it.

\begin{table}
	\centering
	\renewcommand{\arraystretch}{1.3}
	\begin{tabular*}{\textwidth}{@{\extracolsep{\fill}}llllc@{}}
		\toprule
		Algorithm & \multicolumn{1}{c}{$N$} & \multicolumn{1}{c}{$\varepsilon$} & \multicolumn{1}{c}{$\delta$} & Mergeable? \\
		\midrule
		ReqSketch \cite{ReqSketch} & $\log^{1.5}(N)$ & $\varepsilon^{-1}$ & $\sqrt{\log(\delta^{-1})}$ & \cmark \\
		\cite{GribelyukSWY25} & $\log(N)\cdot\log\log(N)$ & $\varepsilon^{-1}\log^2(\varepsilon^{-1})$ & $\log^3\log(\delta^{-1})$ & \xmark \\
		Ours & $\log(N)\cdot\log\log(N)$ & $\varepsilon^{-1}\log(\varepsilon^{-1})$ & $\log^2\log(\delta^{-1})$ & \cmark \\
		\bottomrule
	\end{tabular*}
	\caption{Comparison of the main randomized streaming algorithms for quantile estimation with relative error -- the dependence of the space on relevant parameters and whether the algorithm is mergeable; $N$ is the number of input items;
	$\varepsilon$ is the relative error; 
	$\delta$ is the probability of not satisfying the error bound.
	}
	\label{comparison_space}
\end{table}

\paragraph{Comparison to lower bounds and deterministic bounds}
First, there is a straightforward information-theoretic lower bound of 
$\Omega(\varepsilon^{-1}\cdot \log(\varepsilon N))$ memory words for any relative-error quantile summary, even if constructed offline (i.e., with random access to the whole input);
see Theorem~2 in~\cite{CormodeKMS05} or Appendix~A in~\cite{ReqSketch}.

In the streaming setting, a lower bound\footnote{
	The bound holds under the (reasonable) assumption that $\varepsilon \in \Omega(N^{-c})$ for some constant $c < 1$. Without the assumption, the bound gets smaller by a factor of $\log(\varepsilon N)/\log N$.} 
of $\Omega(\varepsilon^{-1}\cdot \log(\varepsilon N)\cdot \log \log \delta^{-1})$
for randomized comparison-based algorithms follows from Theorem~6.5 by Cormode and Veselý~\cite{add_lower}.
In particular, they show a lower bound of $\Omega(\varepsilon^{-1}\log^2(\varepsilon N))$ for deterministic comparison-based algorithms. If there were a randomized comparison-based algorithm with space $o(\varepsilon^{-1}\cdot \log(\varepsilon N)\cdot \log \log \delta^{-1})$ and we set $\delta < N!^{-1}$, by a union bound, there would be a sequence of random bits that gives low enough error for all possible input orders, yielding the existence of a deterministic comparison-based algorithm violating the deterministic lower bound.  Compared to this lower bound, the additional space factor of our algorithm is only $\O(\log \log \delta^{-1} + \log \varepsilon^{-1} + \log \log(\varepsilon N))$.

The method above shows that a deterministic algorithm can be constructed from a randomized one by using an extremely small $\delta$ and then guessing and fixing the right sequence of (previously random) bits. However, if the number $r$ of random bits used by the algorithm is polynomial in $\varepsilon N$ (which is true for many randomized algorithms, including ours), then there is a much simpler method without guessing the right random bits: If we just set $\delta < 2^{-r}$, then by definition, \emph{all} choices of random bits must give a small enough error for each input and we immediately
obtain a correct deterministic algorithm by choosing arbitrarily instead of randomly.

Thus, by this method, our main theorem with $\delta < 2^{-\varepsilon N}$ yields a deterministic comparison-based algorithm with space $\O(\varepsilon^{-1} \log^3(\varepsilon N))$. This algorithm needs a polynomial upper bound on $N$ to set some fixed parameters based on the value of $\log \ln \delta^{-1}\in \O(\log(\varepsilon N))$. This requirement can be lifted in the streaming setting by a standard trick in which we maintain a sequence of sketches of exponentially increasing sizes with different parameters; see e.g.~Section~5 in~\cite{ReqSketch}. 

With these modifications, the resulting algorithm matches exactly the space bound of the state-of-the-art deterministic algorithm by Zhang and Wang~\cite{zhang}, which achieves the same space bound and is mergeable only with a polynomial upper bound on $N$. Thus, any further improvement of the dependence of the space of our algorithm on $\delta$ would yield new results also for deterministic algorithms -- by the first method existentially and by the second method even constructively, but assuming a polynomial number of random bits. It is an open problem whether this is possible, as there is a $\log(\varepsilon N)$ gap between the known lower and upper bounds for deterministic comparison-based algorithms~\cite{add_lower}.

\subsection{Organization of the paper}
We have presented our result and compared it to the state-of-the-art sketches and known lower bounds.
In the remainder of \Cref{intro}, we briefly comment on practical performance (\Cref{practical}), review further related work (\Cref{prior_work}), and introduce basic notation (\Cref{notation}). 
In \Cref{technical_overview}, we explain the intuition behind our algorithm and compare it with previous work in~\cite{ReqSketch,GribelyukSWY25}.
Then, in \Cref{description}, we describe our sketch in detail and state a few crucial observations. \Cref{analysis} provides an analysis of the space bound and of the error guarantee.
We conclude with a discussion of the main problem in \Cref{conclusions}.

\subsection{Practical performance}\label{practical}
Despite a somewhat intricate analysis of the space requirements,
our algorithm is suitable for implementation due to its simplicity. To support this claim, we prepared a simple proof-of-concept implementation of the sketch and published it on GitHub\footnote{\url{https://github.com/domestomas/adaptive-sketch-experiments}}. The algorithm is indeed under 300 lines of Python code, including the merge operation, comments, functions for approximating ranks and quantiles, etc. We include a simple testing framework for comparing the error with ReqSketch~\cite{ReqSketch}, using the proof-of-concept implementation by the authors\footnote{
	\url{https://github.com/edoliberty/streaming-quantiles/blob/master/relativeErrorSketch.py}}. 
Early experiments on sorted, reverse-sorted, and random streams of length up to $10^9$ gave similar results in the streaming model, while our algorithm provides approximately half the error when constructing the sketch via merging sketches summarizing a similar number of items. From these early results, we conclude that our algorithm is suitable for implementation and especially promising when merge operations are needed. After the internal parameters are fine-tuned for practice, we believe that our algorithm will actually outperform ReqSketch~\cite{ReqSketch} in practice; however, this is outside the scope of this primarily theoretical paper, and we leave it to future work.

\subsection{Further related work}\label{prior_work}
Here, we focus on state-of-the-art quantile sketches with various error guarantees, summarizing their space bounds and mergeability properties.
The space bounds are given in memory words with $\Theta(\log\abs{\mathcal{U}} + \log N)$ bits, where $N$ is the input size and $\mathcal{U}$ is the universe of input items.
For randomized algorithms, we present the space bounds assuming a constant probability $\delta$ of an error that is too large for any fixed query.

\paragraph{Additive error}

Besides relative error, there is a large body of work on quantile summaries with the \emph{additive (uniform) error} of $\pm \varepsilon N$,
i.e., without higher accuracy for the tails.
The state-of-the-art additive-error sketch is the randomized comparison-based KLL sketch by Karnin, Lang, and Liberty~\cite{KLL}, achieving full mergeability in the optimal space $\Theta(\varepsilon^{-1})$.
For deterministic comparison-based sketches, the optimal space of $\Theta(\varepsilon^{-1}\cdot \log(\varepsilon N))$ is achieved by the Greenwald--Khanna (GK) summary~\cite{GK}; its optimality has only been shown recently~\cite{add_lower} (for a simplified version of the GK summary, see~\cite{GK_new}). However, the Greenwald--Khanna summary is not known to be fully mergeable, only ``one-way mergeable'' as discussed in~\cite{agarwal2013mergeable}. An older deterministic algorithm by Manku, Rajagopalan, and Lindsay~\cite{MRL} achieves full mergeability in space $\O(\varepsilon^{-1}\cdot \log^2(\varepsilon N))$.

The non-comparison-based deterministic additive-error $q$-digest~\cite{q_digest} achieves full mergeability in space $\O(\varepsilon^{-1}\cdot \log(\abs{\mathcal{U}}))$, that is, depending on the universe size, which must be known in advance. In this model, Gupta, Singhal, and Wu~\cite{GuptaSW-noncomparison_optimal} recently achieved the optimal space $\Theta(\varepsilon^{-1})$ deterministically, using a compressed version of $q$-digest, but they did not prove that their sketch is fully mergeable.

\paragraph{Relative error}
Apart from the work in~\cite{zhang,ReqSketch,GribelyukSWY25} mentioned above,
there is a modified version of the aforementioned $q$-digest by Cormode, Korn, Muthukrishnan, and Srivastava~\cite{q_digest_new} achieving relative error in space $\O(\varepsilon^{-1}\cdot\log(\varepsilon N)\cdot \log(\abs{\mathcal{U}}))$. Like the original $q$-digest, it is deterministic, non-comparison-based, fully mergeable, and requires prior knowledge of the universe $\mathcal{U}$.

\paragraph{Dynamic streams}
When the input items may be deleted, the currently best quantile summary is Dyadic CountSketch~\cite{luo2016quantiles},
requiring space $\O(\varepsilon^{-1}\cdot\log^{1.5}(\abs{\mathcal{U}})\cdot \log^{1.5}(\log(\abs{\mathcal{U}}) / \varepsilon))$;
very recently, the algorithm has been slightly simplified and its space usage improved to $\O(\varepsilon^{-1}\cdot\log^{1.5}(\abs{\mathcal{U}}))$~\cite{bulanek2026matrixfactorizationapproachturnstile}.
This algorithm, as well as any other which can handle arbitrarily many deletions,
is non-comparison-based. Furthermore, no sublinear-space algorithm can provide relative error guarantees in dynamic streams, since otherwise
one can first insert $n$ arbitrary elements and then repeatedly query for the minimum and remove it,
which retrieves all of the insertions as the minimum must be returned exactly under relative error.

\paragraph{Other quantile sketches}

Finally, besides algorithms designed with theoretical guarantees in mind, it is worth mentioning that many practitioners actually use $t$-digest~\cite{tdigest}, which is usually highly accurate in practice and supports merging. While $t$-digest aims at uniform or relative error depending on its parameter setup, it provides no worst-case guarantees, and indeed, there are adversarially constructed inputs, or even samples from certain distributions, that make the error of $t$-digest almost arbitrarily large~\cite{ReqSketch_improvements}.
A recently proposed SplineSketch~\cite{splineSketch} achieves better accuracy in practice than $t$-digest while also providing close-to-optimal uniform error guarantees, based on a suitable parameterization of the input.

\subsection{Preliminaries}\label{notation}

Here, we outline the notation and terminology used throughout the paper:
\begin{itemize}
\item By \emph{item}, we always mean an arbitrary item from a universe $\mathcal{U}$ with a total order.
For simplicity, we assume throughout the paper that all the input items are different.
However, all the definitions can be extended to support equal items, and all the
algorithms and proofs are correct even without this assumption.

\item An input $\mathcal{I}$ consists of items from $\mathcal{U}$. We denote the input size by $\abs{\mathcal{I}} = N$.

\item For any set of items $M$, let $\rank(y, M)$ be the rank of item $y$ in $M$, and let $\rank(y) = \rank(y, \mathcal{I})$ be the rank of $y$ in the input.

\item We use $\estrank(y)$ for the answer returned by the sketch for a rank query $y$.

\item The error of a rank query $y$ is $\err(y) = \estrank(y) - \rank(y)$.

\item To avoid edge cases, by $\log x$ we always mean $\log_2(\max(2, x))$ and by $\ln x$ the natural logarithm $\log_e(\max(e, x))$.

\item We denote the ranges of arrays as mathematical intervals; thus, $)$ means exclusion and $]$ inclusion. For example, $A[3, 5]$ denotes the elements $A[3], A[4]$ and $A[5]$, while $A[3, 5)$ means only $A[3]$ and $A[4]$.

\end{itemize}

\subsection{AI methodology}

\begin{flushleft}
\renewcommand{\arraystretch}{1.3}
\begin{tabular}{@{}l@{\qquad}c@{}}
\textbf{Activity} & \textbf{AI used} \\
Asking the research question & No \\
Coming up with the approach & No \\
Proof development & No \\
Writing and exposition & No \\
Checking for bugs & Yes \\
Other supporting tasks & Yes \\
\end{tabular}
\end{flushleft}

\begin{sloppypar}
\paragraph{Details (elaborating on the points above).}
The main use of AI was for repeatedly finding bugs in the error analysis, where we originally violated the conditions of the probability bound we used. This forced us to come up with a better proof. We also used it to find some less critical technical erros in the analysis of the space bound, and we let AI repair some of them. This included proving some observations. We let AI check our grammar and generally polish the paper, refine some arguments, etc\dots We personally verified everything and we take full responsibility for the result.
\end{sloppypar}

\section{Technical overview and detailed comparison to prior work}\label{technical_overview}
In this section, we give some intuition behind our algorithm and compare it to ReqSketch~\cite{ReqSketch}
and the streaming algorithm from~\cite{GribelyukSWY25}.
This section can nevertheless be skipped as our algorithm's description and analysis do not rely on it.
To better present the core ideas, we will make many simplifications and allow for some imprecision, typically in constant factors.
In particular, we ignore the parameter $\delta$, which is needed for high-probability bounds. For the formal description and analysis, see \Cref{description,analysis}.

After a high-level overview of the basic sketch design in \Cref{overview_basic}, which stems from the KLL sketch~\cite{KLL} for the uniform error, we explain our sketch in two stages by comparing it to the two main results we build upon. In \Cref{overview_compactors}, we describe the \emph{relative compactors} from~\cite{ReqSketch} and explain our \emph{adaptive compactors}. In \Cref{overview_elastic}, we give a quick overview of the streaming algorithm from~\cite{GribelyukSWY25} and describe the rest of our algorithm, in particular, the notion of \emph{marking} and \emph{non-marking} compactions and the \emph{space allocation strategy}. 

\subsection{Basic sketch design}\label{overview_basic}
The basic sketch design originally comes from the KLL sketch~\cite{KLL}. KLL consists of a series of \emph{compactors}, which are essentially buffers, arranged into \emph{levels}. Items from the input stream are added to the level-0 compactor and, when any compactor reaches its \emph{capacity} (it is \emph{full}), we perform the \emph{compaction} operation.
The compaction sorts the items in the compactor, deletes odd-indexed or even-indexed items with equal probability, and moves the rest to the compactor one level higher. Therefore, each item at level $h$ ``represents'' $2^h$ items of the original stream.
Then the rank estimate for a query item $y$ is simply the sum of weights of all stored items $y'\le y$.
Moreover, as the items are sorted and we delete each odd/even-indexed item, each promoted item represents a deleted item of a similar rank.
Each compaction affects the error of rank estimates for some universe items, namely, by either increasing or decreasing their rank estimate by $2^h$ with equal probability, where $h$ is the level. This implies that the expected error change is zero, and then the error analysis for any item $y$ mainly boils down to bounding the number of compactions affecting the estimation error for $y$.

In our sketch, we use novel \emph{adaptive compactors} together with a new \emph{space allocation strategy} to determine the capacity of each compactor.

\subsection{Relative vs.\ adaptive compactors}\label{overview_compactors}

In this section, we briefly sketch the \emph{relative compactors} from~\cite{ReqSketch}, also used in~\cite{GribelyukSWY25}, and then explain the design of our \emph{adaptive compactors}.

\subsubsection{ReqSketch and relative compactors}
In the \emph{relative compactors} used in ReqSketch~\cite{ReqSketch} (and also in our \emph{adaptive compactors}), the compaction operation is performed only on the $X$ largest items of the compactor, for some even $X$; see \Cref{compaction_general}.
This is based on the observation that if $y$ is smaller than all the items participating in the compaction, the compaction does not affect the error for $y$.
The items participating in the compaction are determined in such a way that the smallest $\varepsilon^{-1}$ items are never removed from the compactor. This is related to the relative error guarantee, as the accuracy of the sketch must be higher for items of smaller rank. In particular, for the smallest $\varepsilon^{-1}$ items, the error must equal zero, and so we must store them all (and indeed, the level-0 compactor always stores the smallest $\varepsilon^{-1}$ inserted items).
In other words, the items of rank roughly $\varepsilon^{-1}2^h$ are summarized by the level $h$ of the sketch.

\begin{figure}
\centering
\includegraphics[width=\textwidth]{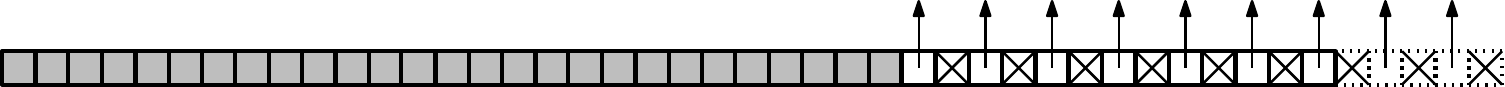}
\caption{
	The compaction operation of relative and adaptive compactors happens when the compactor is full (all the non-dashed positions contain an item). 
	Items are first sorted from the smallest to the largest.
	The compaction evicts the crossed items from the memory and ``promotes'' the items with arrows to the next level,
	while the gray ones do not participate in the compaction, so they remain in the buffer.
	The number of items participating in the compaction is, however, computed differently in relative and adaptive compactors.}\label{compaction_general}
\end{figure}

A crucial part of the algorithm is the choice of the number of items participating in the compaction. The authors in~\cite{ReqSketch} use a deterministic process called the \emph{compaction schedule}. The compactor is split into $S+1$ \emph{sections} of size $K$, where $(S+1)\cdot K$ is the \emph{capacity}, and the schedule is determined by the binary representation of the number $P$ of already performed compactions; namely, if $Z$ is the number of trailing ones of $P$, the sections $[S-Z+1, S+1]$ participate in the compaction.

Let us explain this in detail.
Given the number $P$ of compactions already performed on this level, the section size $K$, and the section count $S$, the compaction happens when there are at least $K\cdot (S+1)$ items in the buffer (the compactor is full).
We index the sections and bits of the binary representation of $P$ starting from index $1$ in such a way that the smallest items are in section $1$,
the most significant bit of $P$ has index $1$, and 
the least significant bit has index $S$.
Let the $i$-th bit of $P$ correspond to the $i$-th section of the compactor. We say that section 
$j \leq S$ is \emph{marked} if the $j$-th bit is set to 1. 
Letting $Z$ be the number of trailing ones of the binary representation of $P$,
the items from section $S+1$ and from the $Z$ marked sections $[S-Z+1, S]$ participate in the compaction. As in KLL, the odd- or even-indexed participating items (with equal probability) are promoted to the next level, and the rest are discarded. Then we increment $P$, which essentially means marking the previously unmarked section $s = S-Z$ and unmarking sections $[S-Z+1, S]$; see \Cref{compaction_relative}.

\begin{figure}
\centering
\includegraphics[width=0.9\textwidth]{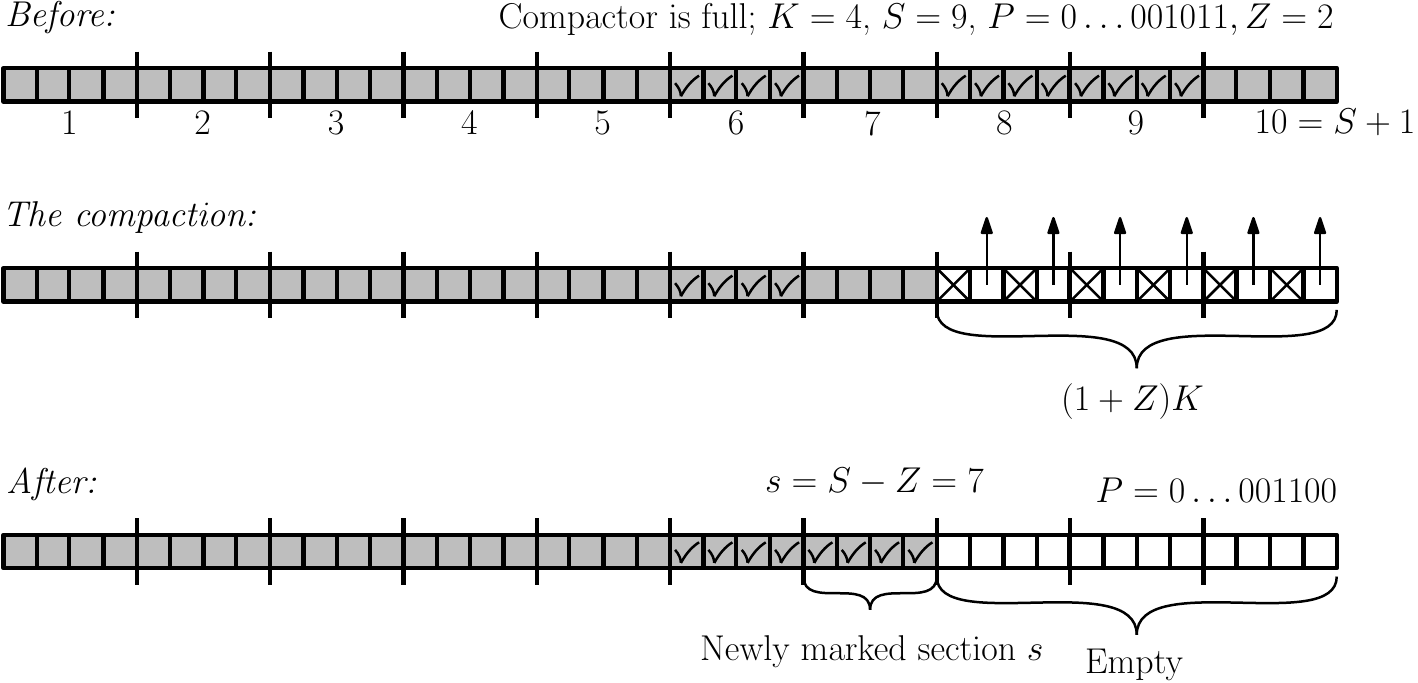}
\caption{The compaction operation in a relative compactor;
$P$ is the number of already performed compactions (written in binary), and $Z$ is the number of trailing ones of $P$.}\label{compaction_relative}
\end{figure}

To see the motivation behind the compaction schedule, observe that at the time $t_0$ when section $s$ is marked, all items in $s$ are of smaller rank than all items participating in the compaction. After another $2^Z$ compactions, at some time $t_1$, the marked section $s$ is compacted.
Generally, the items present in section $s$ at time $t_1$ can be different from the items present at time $t_0$.
However, due to the sorting before each compaction, all the new items have ranks smaller than or equal to those of all the old ones. This implies that \emph{for each compaction $p$, there is a set $U_p$ of $K$ distinct items of rank smaller than all items participating in $p$, which are either later removed by another compaction, or remain in the compactor at the end of the algorithm;
crucially, the sets $U_p$ are pairwise disjoint.} Let us call this the \emph{schedule property}.

Why is this so important? Let us recall the observation that if $y$ is smaller than all the items participating in compaction $p$, then $p$ does not affect the error for $y$. For a fixed item $y$ and level $h$, let $R_h$ be the number of items smaller than $y$ ever inserted into $h$. Then, by the schedule property, the number of level-$h$ compactions affecting the error of $y$ is at most $R_h/K$. This helps us bound the overall error for $y$.

As for the space complexity,
observe that for the schedule property, we need the section count to be large enough; in particular, we need $S \geq \log(P)$, so that there is always an unmarked section to be marked after the compaction and $P$ never overflows.
As the number $P$ of compactions is in $\O(N/K)$, it suffices to set $S \in \Omega(\log(N/K))$. 
In ReqSketch, they choose $K = \varepsilon^{-1} / \sqrt{\log (\varepsilon N)}$ and thus each compactor contains asymptotically $SK \leq \varepsilon^{-1}\sqrt{\log (\varepsilon N)}$ items, while the number of compactors is $\log(\varepsilon N)$, resulting in space complexity $\O(\varepsilon^{-1}\log^{1.5}(\varepsilon N))$. Both in~\cite{GribelyukSWY25} and in our result, we use $K = \varepsilon^{-1}$ and allow some of our compactors to contain $\varepsilon^{-1} \log(\varepsilon N)$ items.

Unfortunately, the notion of compaction schedule becomes less elegant under merge operations. The authors in~\cite{ReqSketch} developed a merging algorithm that has two deficiencies:
First, when merging two relative compactors, they take the binary OR of their compaction schedules,
which is efficient, but the nice intuition behind the compaction schedule is somehow lost, and the previously easy-to-see bounds on the number of compactions affecting the error require an involved charging argument (Lemma~6.4 in~\cite{ReqSketch}).
Second, the parameters of each sketch are recomputed multiple times, every time the summarized input size of the sketch gets squared,
and this yields a complicated tight analysis of the variance (e.g., Lemma~6.12 in~\cite{ReqSketch}).

\subsubsection{Adaptive compactors}

We propose an alternative to the compaction schedule, which is a bit less elegant in the streaming setting, but naturally covers full mergeability with much simpler proofs. Most importantly, the adaptive compactor has the property of improving its state after the addition of small enough items.
The core of our method is to tie the marks not to the sections of the compactor, but directly to the items.
Then we can simply merge two compactors by merging the two sorted arrays -- all the marked items simply stay marked. Moreover, when new small items arrive, they push the larger marked items to the right, improving the state of the compactor.

However, in this model, the marked items no longer form contiguous sections of the compactor due to future item additions.
It is possible to analyze the error without enforcing contiguous sections of marked items, but the resulting compaction algorithm and its analysis are more complicated, requiring the notions of dense and sparse sections with respect to the marking.

Instead, we further modify the algorithm to recover contiguous sections of marked and unmarked items, which yields a simpler analysis.
Suppose the lowest-ranked item removed during a particular compaction has rank $r_0$ and the items marked after the compaction have rank $r_K < r_{K-1} < \cdots < r_1$. Then a new item with rank $r_x$ arrives such that $r_K < r_x < r_0$. Such an item can surely be marked instead of the item of rank $r_K$. With such remarking, the $K$ items marked for a particular compaction always form a contiguous section in the compactor. To do this remarking efficiently, we do not actually mark the items, but remember for each compaction the smallest-ranked removed item, called a \emph{marker}. From this information, we can recover the marking of items in linear time. To further simplify the analysis, we perform the ``remarking'' in such a way that the contiguous sections of marked items stay aligned, i.e., they start at positions divisible by $K$, which allows us to have marked sections instead of marked items, similarly to relative compactors.

Given the marking of sections, determining the size of the compaction is conceptually the same as for relative compactors. We look for the rightmost unmarked section $s$, remove all items to the right of this section together with their markers, and then mark $s$, as depicted in \Cref{compaction_relative}.

Our strategy not only preserves the schedule property, but for adaptive compactors, this property is a simple observation, even under full mergeability, while the proof for relative compactors is quite involved.

Replacing relative compactors by adaptive compactors in ReqSketch~\cite{ReqSketch} yields the same space and error bounds with a much simpler proof of the error bound, while the proof of the space bound gets more involved. 
However, we use the additional property of improving the state by the addition of small items as a substitute for the resets introduced in~\cite{GribelyukSWY25} to obtain almost optimal space while preserving the full mergeability. To do this, we need an alternative to the \emph{space allocation strategy} from~\cite{GribelyukSWY25} to assign the right capacities to the compactors, and one additional trick which we call \emph{skipping levels}.

\subsection{Space allocation strategy \& friends}\label{overview_elastic}
In this section, we first give a high-level overview of~\cite{GribelyukSWY25}. Then, we explain the rest of our sketch in comparison to their streaming algorithm. For a quick comparison of the compactors and space allocation strategies used in the three sketches, see \Cref{comparison_compactors}.

\begin{table}
	\centering
	\small
	\renewcommand{\arraystretch}{1.3}
	\begin{tabularx}{\textwidth}{@{}l>{\raggedright\arraybackslash\hsize=0.6\hsize}X>{\raggedright\arraybackslash\hsize=1.4\hsize}X@{}}
		\toprule
		Algorithm & \multicolumn{1}{c}{Compactors} & \multicolumn{1}{c@{}}{Space allocation strategy} \\
		\midrule
		ReqSketch~\cite{ReqSketch} & Relative compactors & One compactor per level, all of the same size, growing in time \\
		\addlinespace
		\cite{GribelyukSWY25} & Elastic compactors (relative compactors that are dynamically resized) & Each level is a subsketch with multiple elastic compactors; levels satisfy the monotonicity property; dynamic space allocation based on resets and the schedule property \\
		\addlinespace
		Ours & Adaptive compactors & One compactor per level; dynamic space allocation based on counters that express the state of the compactors \\
		\bottomrule
	\end{tabularx}
	\caption{Comparison of the compactors and space allocation strategies used in the three main sketches.}
	\label{comparison_compactors}
\end{table}

\subsubsection{The algorithm by Gribelyuk, Sawettamalya, Wu, and Yu}
First, we stress that the following is a very simplified overview of~\cite{GribelyukSWY25}. Keeping this in mind, let us briefly outline their algorithm.

Their sketch again consists of $\O(\log(\varepsilon N))$ levels, each summarizing items of rank roughly $\rank_h \eqdef \varepsilon^{-1} 2^h$ w.r.t.\ the current input stream (as in ReqSketch and in our algorithm).
However, in their algorithm, each level contains a subsketch consisting of a sampler and $\log(\varepsilon^{-1})$ ``sublevels'', each containing one compactor. 
The items on sublevel $i$ of the subsketch on level $h$ of the sketch have weight $\lceil \varepsilon 2^{i+h} \rceil$.

\paragraph{The monotonicity}
For all levels $h>0$, all items in the subsketch on level $h$ are larger than all items on level $h-1$. Let us call this property the \emph{monotonicity} of the sketch. 
To keep the sketch monotone, they do not insert all new items into level $0$, but for each incoming item $y$, they choose the only level where monotonicity is preserved. 
Before an item is inserted into a level $h$, it is sampled by a uniform sampler with probability $\lceil \varepsilon 2^h \rceil^{-1} \approx \rank_h^{-1}\varepsilon^{-2}$, and if sampled, it is inserted into sublevel $0$ of the level-$h$ subsketch.
This gives us the correct weight after insertion. The sampling does not change the asymptotics of the error, as the error due to the sampling is approximately $\varepsilon \rank_h$; multiplying the sampling probability by $\log \delta^{-1}$ gives this error with high probability.

\paragraph{The reset}
When the cumulative weight of all items on some level $h$ is at least $3\cdot \varepsilon^{-1} 2^{h}$, the largest third of them are moved to the subsketch on level $h+1$ (more precisely, items are moved from the largest until the total weight on the level becomes at most $2R_h$). 
They call this situation a \emph{reset}.
The crucial observation is that due to the monotonicity of the sketch, after a constant number of resets, all the items originally on level $h+1$ are ``pushed away'' to the next level.  
To explain the motivation for resets in detail, we need to introduce the \emph{elastic compactors}, an extension of relative compactors.

\paragraph{The elastic compactors}
In the subsketches, they use a version of relative compactors, where the section size is fixed to $K \approx \varepsilon^{-1}$, and the section count $S$ (and consequently the capacity) is given by a \emph{space allocation strategy}. They call this version \emph{elastic compactors} as they can shrink or expand. 
They give a slightly different view of the compaction operation. Instead of the number of compactions, let the sections of the compactor represent a binary counter $T<1$ such that the $i$-th section represents its bit of weight $2^{-i}$. 
The compaction algorithm itself is the same -- given the section count $S$, if section $S+1$ is full, we find the rightmost unmarked section $s$, compact the sections to the right of $s$, and mark $s$. Thus, any compaction increases the counter $T$ by $2^{-S}$ where $S$ is the current section count. 
To preserve the schedule property, we need $\sum_{t \in X} 2^{-S_t} < 1$ where $X$ is the set of all times when a compaction happens, and $S_t$ is the section count at time $t$. This condition ensures that the counter $T$ does not overflow, and we always have an unmarked section to mark after each compaction. They say that such a choice of the values of $S_h$ is \emph{feasible}.

\paragraph{The schedule property}
Finally, the \emph{reset} explained above means that we can reset the counter $T$ to 0. 
Let us recall the schedule property: for each compaction $p$, we need a set $U_p$ of $K$ small items, and this can be the items removed from section $s$ marked by $p$ in the next compaction where $s$ participates. Thus, if after a constant number of resets of $T$, all items originally present in the compactor are pushed out, then the schedule property is preserved; even though each item can be used in constantly many different sets $U_p$, this does not hurt the asymptotics.
Then, instead of the sum in the definition of feasibility above, we only need $\forall i, j: \sum_{t \in W_{i, j}} 2^{-S_{i,t}} < 1$, where $W_{i, j}$ is the set of all times between the $j$-th and $(j+1)$-th reset on level $i$ when a compaction happens (let us consider time 0 and the current time also as times of reset) and $S_{i, t}$ is the section count on level $i$ at time $t$.

\paragraph{The space allocation strategy}
It remains to explain the space allocation strategy.
For each level $i$, let $t_{i, j}$ be the time of the $j$-th reset. For simplicity, considering level $i+1$, let us assume that each time interval $[t_{i+1, j}, t_{i+1, j+1})$ is contained in some interval $[t_{i, \mathrm{parent}(j)}, t_{i, \mathrm{parent}(j)+1})$ (the general case needs some technical changes as explained in the original paper). Then, the section count of any compactor of a subsketch on level $i$ at time $t\in [t_{i+1, j}, t_{i+1, j+1})$ is given as 
$$
S_{i, t} = \log\left(\frac{t_{i, \mathrm{parent}(j)+1} - t_{i, \mathrm{parent}(j)}}{t_{i+1, j+1} - t_{i+1, j}}\right) + 5\log \varepsilon^{-1},
$$
where the time is measured in changes to the memory of the whole sketch.
As time is polynomial in $N$, the sum over all levels telescopes to $\O(\log N \cdot \log \varepsilon^{-1})$ at any given time $t$, and as we have $\log \varepsilon^{-1}$ compactors per level, it yields near-optimal space. Moreover, within any interval $[t_{i+1, j}, t_{i+1, j+1}]$, there can be at most $3\varepsilon^{-2}$ insertions into the level $i$ by the definition of the reset. Thus, during each ``child'' interval on level $i+1$, we have at most $3\varepsilon^{-2}$ compactions on any fixed compactor of level $i$ of the sketch. It follows that for child intervals $[t_{i+1, m}, t_{i+1, m+1}], \dots, [t_{i+1, r-1}, t_{i+1, r}]$ of parent interval $[t_{i, k}, t_{i, k+1}]$ we have
$$
\sum_{t \in W_{i, k}} 2^{-S_{i, t}} \leq \sum_{q=m}^{r-1} 3\varepsilon^{-2} \cdot \frac{t_{i+1, q+1} - t_{i+1, q}}{t_{i, k+1} - t_{i, k}} \varepsilon^{5} 
\leq 0.25 \cdot  \frac{\sum_{q=m}^{r-1} (t_{i+1, q+1} - t_{i+1, q})}{t_{i, k+1} - t_{i, k}} \leq 0.25,
$$
which gives us feasibility.
Thus, their algorithm preserves both feasibility and the schedule property and consequently the relative error guarantee in space close to $\varepsilon^{-1} \log N$.

\paragraph{Limitations}
We conclude with shortcomings of the result in~\cite{GribelyukSWY25}.
Apart from the overall complexity of the algorithm's description, it comes with no merge operation, which further limits its practical usability.
Indeed, it is not clear how to merge two elastic compactors so that the error can be bounded, as even for simpler relative 
compactors, the analysis under arbitrary merging is highly involved~\cite{ReqSketch}. 

\subsubsection{Our algorithm}

Having explained~\cite{GribelyukSWY25}, let us now dive into our algorithm. 

\paragraph{Monotonicity} First, observe that our adaptive compactors naturally possess a property similar to the resets used in~\cite{GribelyukSWY25}. Insertion of new, small items into the buffer causes the marked items to move to the right. This yields an even more ``fine-grained'' property -- if we associate a binary counter $T<1$ with the marked sections as in the elastic compactors, inserting $K$ sufficiently small items reduces $T$ by half.

However, to leverage this property, we would need something similar to monotonicity. It turns out the idea of inserting items into the ``right'' level is not directly applicable to our sketch, as we have just one compactor per level, and thus the weight of items on level $h$ is just $2^h$ in our sketch and sampling with probability $2^{-h}$ yields an error that is too large for items of rank roughly $\varepsilon^{-1}2^h$. However, the solution is simple: We insert items of rank roughly $\varepsilon^{-1}2^h$ into level $h - \log \varepsilon^{-1}$ instead of level $h$. Now we can sample the items with probability $2^{-h}\varepsilon^{-1}$. Of course, we do not obtain the pure monotonicity property, but each item $y$ inserted into level $i$ is smaller than all items on levels $\approx i + \log \varepsilon^{-1}$ and higher at the time of insertion of $y$. It turns out this property is sufficient for us.

\paragraph{Skipping levels vs.\ adaptiveness}
Let us explain why it is beneficial to combine skipping low levels and our adaptive compactors. The adaptive compactors work best with a reverse-sorted stream, when each new item is smaller than the previous one. In this scenario, we only need constant space per compactor, obtaining optimal space $\O(\varepsilon^{-1}\log (\varepsilon N))$. On the other hand, for sorted streams, the adaptive compactors (without skipping levels) behave exactly as relative compactors. Using skipping levels but not adaptive compactors, it is the other way around: For a reverse-sorted stream, we always insert into level $0$, and thus no skipping happens. For a sorted stream, skipping levels is equivalent to iteratively building additive-error sketches for prefixes of the stream of length $\varepsilon^{-1}2^i$, for $i = 0, \dots, \log_2(\varepsilon N)$, each with additive error $\varepsilon$. This can be done in the optimal space $\O(\varepsilon^{-1}\log \varepsilon N)$. We see that the adaptivity of our compactors and the notion of skipping levels form a kind of duality, and it is best to let each approach the problem from the opposite direction.

\paragraph{Marking and non-marking compactions}
The use of samplers to skip some low levels works, but they are in fact not necessary, and removing them leads to a cleaner algorithm description.
Let us explain, from another angle, the reason behind using a sampler to skip relatively low levels of the sketch. We claimed above that this is so that the new items are smaller than the items on the next level, and help us to improve the state of the next level, eventually resetting it.

Now, one may ask why not insert all the items into level $0$, as they would eventually get to the appropriate level where they cause the reset/improvement.
This is because all the discussed variants of compactors have some state given by marked sections, which we can represent by some kind of counter. By doing compactions, the state gets somehow ``worse'' as the counter increases, and so does the space usage. Thus, by skipping some unnecessary levels, we save ourselves the damage done on those levels and save some space. This line of reasoning yields a technique of obtaining the same result without using samplers -- in the case of a compaction where only large items participate, we simply do not mark a new section/items and thus do not damage the compactor as we do not increment the counter (here, the items must be large relative to their level).

Let us explain this in more detail. Let $\min(p)$ be the smallest item participating in a compaction $p$ on some level $i$. We check the approximate rank of $\min(p)$, as we did in the case of insertion with samplers: We find the highest level $h$ such that $\min(p)$ is larger than the smallest $\varepsilon^{-1}$ items on $h$. If $h > i + \log \varepsilon^{-1}$, level $i$ is too low for $\min(p)$ and should be skipped. In this scenario, we \emph{do not mark a new section} after compaction $p$. We say that this compaction is \emph{non-marking}, in contrast to \emph{marking} compactions. Non-marking compactions do not increase the counter and do not degrade the state of the compactor. Then, for any item $y$, the error for $y$ from the compactions on the levels in which items $\le y$ participated in non-marking compactions is at most the error of the sampler -- the process itself is in fact analogous to reservoir sampling from the point of view of $y$. Note that this strategy has the additional nice property that if the rank of $y$ in the input stream changes significantly after its insertion, but while $y$ is still on some ``low'' level, the compactions are marking or non-marking with respect to the current rank. Of course, it can still happen that the rank of $y$ gradually increases while $y$ is close to its designated level, and thus, this does not help us in the worst case. Also note that a similar strategy of marking and non-marking compactions can be used in the sketch from~\cite{GribelyukSWY25}, only with level difference $1$ instead of $\log \varepsilon^{-1}$.
We call this level difference the \emph{gap} and denote it by $G$. In the actual algorithm, we need $G$ larger than $\log \varepsilon^{-1}$ due to the high-probability bounds and some technicalities in the proof of the space bound.

\paragraph{Space allocation strategy}
Finally, let us discuss our allocation strategy. The most basic idea is analogous to the one used in~\cite{GribelyukSWY25}. Recall that they give the section count as the difference of logarithms of the amount of time from the last reset for the current level and for the level above. The logarithm of the time from the last reset can be viewed as the state of the compactor -- the longer from the last reset, the worse the state. Thus, the section count is very roughly given by ``how much worse the state of this level is compared to the level above''. That is the general idea we want to use.

Let us start with a high-level overview of our strategy. Each compactor has some state given by its marked sections, represented by some real counter $T$. Our counter is technically more complicated than the previous versions, but still, a larger $T$ corresponds to a worse state. Non-marking compactions do not increase $T$. A marking compaction $p$ on level $h$, which marked a set of items $X$, increases $T_h$ by something like $2^{-S_h}$, where $S_h$ is the current section count on level $h$. For this increase, we would like to pay by decreasing $T_i$ for some level $i > h$ due to the items from the set $X$ being inserted into the level-$i$ compactor and decreasing the counter $T_i$ by shifting the marked items/sections to the right. We know that as $p$ is marking, almost all of the items on levels $h+G$ and higher are larger than the items in $X$ -- thus inserting them into any of those levels would cause a significant decrease in $T$. The core of our strategy is that, over all levels $i > h+G$ (where $G$ is the gap), we pick level $b(h)$ such that $T_{b(h)}$ is maximal; we call it the \emph{bounding level}. Then we set the section count of $h$ to something like $T_h - T_{b(h)}$. In other words, among the levels that would improve their state by inserting the items of $X$, we pick the one where the improvement is the largest (as the improvement is proportional to the counter). We set the section count on level $h$ such that the damage done by the compaction $p$ cancels with this improvement. The larger $T_{b(h)}$ is, the larger the improvement, the larger the damage we can allow, and the smaller $S_h$ can be used.

\paragraph{Key properties of the strategy}
How can we prove the space bound of this strategy? We use a telescoping property similar to the one in the strategy in~\cite{GribelyukSWY25}.
First, consider the case $G = 1$. Some of the levels have constant section count; those take cumulative space $\O(\varepsilon^{-1} \log {\varepsilon N})$ (as we have $K \approx \varepsilon^{-1}$). For the set $V$ of levels with non-constant section count, the sum $\sum_{h \in V} (T_h - T_{b(h)})$ telescopes and the cumulative space is $\O(\varepsilon^{-1} T_{\max})$. Thus, we have a good space bound assuming $T_h \in \O(\log(\varepsilon N))$ for any $h$. If we have $G>1$, the same telescoping argument gives us a $G$ times larger space. Thus, the proof boils down to having $T$ bounded, which is the main challenge.

Regarding the schedule property, we need the feasibility as before -- there must always be an unmarked section in the compactor so that $T$ does not ``overflow''. It again boils down to the properties of $T$, and thus it is time to properly define the counter $T$ and give a more precise formula for the section count.

\paragraph{Our counters}
Let us call the leftmost unmarked section the \emph{pivot} section and let $x$ be the number of sections to the left of the pivot. Let those sections to the left of the pivot have weights $2^0, 2^1, \dots, 2^{x-1}$ and let the rest of the sections, including the pivot, have weights $2^{x}, 2^{x-1}, \dots$. The counter $T$ is simply the sum of the weights of the marked sections. 
We define the section count $S$ as $2x+1 + \max(0, T_h - T_{b(h)})$.
Thus, the compactor always contains all the sections of weight at least 1 (there are exactly $2x+1$ of them) and on top of them another $\max(0, T_h - T_{b(h)})$ sections.

This representation of $T$ and $S$ has the properties we need. First, by definition, there is always an unmarked section -- the pivot section. This gives us feasibility. Second, the first ``mandatory'' part of the compactor has $2x + 1 \leq 2\log(T) + 1$ sections. This is not constant, but assuming $T \in \O(\log(\varepsilon N))$, it is $\O(\log\log(\varepsilon N))$ and thus it increases the space bound only by a factor of $\O(\log\log(\varepsilon N))$.

And third, a marking compaction increases the counter by $2^{-\max(0, T_h - T_{b(h)})}$, which is exactly what we need to formally prove the core idea above -- that for the increase of $T_h$ during a compaction we pay by the decrease of $T_{b(h)}$ later. We get this by defining a potential function $\Phi$ roughly as the sum of $\Phi_h = 2^{4h+T_h}$ over all levels $h$ and proving that the potential is polynomial in $\varepsilon N$. The proof is a bit technical, but the most basic idea is the one explained above -- the increase of $\Phi_h$ during a compaction is bounded by the decrease of $\Phi_{b(h)}$ later. The term exponential in $h$ in the definition of $\Phi$ is there to compensate for the fact that the number of small items moving up decreases approximately by half at each level (now, the number of items decreases by a factor of two, but their ability to destroy some $\Phi$ above increases $2^4$ times). 
The only case where we cannot charge the increase of $\Phi_h$ to a future decrease of $\Phi_{b(h)}$ is when all levels above $h+G$ are in a good state (in particular, all have $T<2$). In this case, we charge the increase to an imaginary level $\infty$, and we prove that the amount of potential charged to that level throughout the whole algorithm is polynomial in $\varepsilon N$. For a more detailed explanation, see \Cref{space_intuition}.

\paragraph{Mergeability}
Let us conclude this section with some notes regarding mergeability. Most of the ideas presented above also work in the mergeability setting; however, the definition of the counter $T$ does not. In particular, merging two compactors with counters $T_1, T_2$ can result in a compactor with counter $T \approx T_1 \cdot T_2$. Thus, in the actual algorithm, we use a different definition of $T$. Let us very briefly sketch the differences; see \Cref{counters} for the full definition. The most important change is that we always keep the leftmost two sections unmarked; they have no weight, and we redefine the pivot section as the third unmarked section from the left. This change alone gives us $T \leq 2\cdot \max(T_1, T_2)$, as the number $x$ of marked sections to the left of the pivot cannot increase by a merge. We make additional changes such that $T \leq \max(T_1, T_2)$ (in fact, we even have a slightly stronger property; see \Cref{obs:tmerge}). 
Our goal is to define $T$ in such a way that any suffix of the sections starting with two unmarked sections has weight zero -- this, together with the previous modification, implies $T \leq \max(T_1, T_2)$ after a merge. For this, in addition to positive weights for marked sections, we give negative weights to unmarked sections, and we define $T$ recursively by its suffixes -- the weight of each suffix is the maximum of 0 and the sum of the weight of its first section and the weight of the suffix one section shorter. For this idea to work, we need a final technical change -- the weights of the sections must increase by a factor $< 2$; in particular, the weight of a marked section is $2^{i/2}$ instead of $2^i$, and the weight of the same section if unmarked is $-\sqrt{2}\cdot 2^{i/2}$.

\section{Description of the algorithm}\label{description}
In this section, we give a full formal description of our algorithm.

The high-level design of the sketch is analogous to KLL~\cite{KLL} or ReqSketch~\cite{ReqSketch}, namely that it consists of compactors arranged into levels. Our main contribution lies in using adaptive compactors described in \Cref{the_compactor} together with an appropriate space allocation strategy described in \Cref{space_allocation}.

\subsection{Basic sketch design}
The sketch consists of a sequence of $H$ \emph{adaptive compactors} indexed from 0 to $H - 1$. We imagine the compactors as arranged in \emph{levels} -- the compactor at level 0 is at the bottom, and the compactor at level $H-1$ is at the top. To simplify some later definitions and proofs, we have one additional \emph{imaginary level} $\infty$ which has no compactor (so it does not contain any items), and it is in fact not represented in the algorithm in any way.

Apart from the sequence of compactors, the sketch consists of the error bound $0 < \varepsilon \leq 1$, the failure probability $0 < \delta \leq 2^{-3}$, the \emph{gap} $G$ (defined in \Cref{space_allocation}), and the section size
\begin{equation}\label{kdef}
K \eqdef 2^5 \cdot \lceil \varepsilon^{-1} \rceil \cdot \lceil \log \ln \delta^{-1}\rceil.
\end{equation}

An adaptive compactor contains a set $B$ of items -- \emph{the buffer} -- with some \emph{capacity} $C$.
Any item received from the input stream is simply added to the buffer of the level-$0$ compactor. Whenever for any compactor we have $\abs{B} \geq C$ (we say that the compactor is \emph{full}), it performs a \emph{compaction operation}, which removes some items from $B$; half of them are evicted from memory, and the other half are added to the buffer of the next-level compactor. If the next-level compactor does not exist, the number of levels $H$ is increased by one, and a new compactor on the highest level is created to receive the items.

\begin{rem}
	By definition, we have $\lceil \log\ln \delta^{-1} \rceil \geq 2$ and $K \geq 2^6$. We use this later in the analysis.
\end{rem}

\subsubsection{Merging two sketches and stream updates}
The merge operation is defined on two sketches $\mathcal{S}_1, \mathcal{S}_2$ with the same parameters $\varepsilon$ and $\delta$ (which we keep unchanged).
Without loss of generality, assume that $H_1 \geq H_2$. For each level $h < H_2$, the compactor on level $h$ of the new sketch is given by a merge of the two compactors on level $h$ from sketches $\mathcal{S}_1, \mathcal{S}_2$. For $H_1 > h \geq H_2$, the resulting compactor is simply the compactor on level $h$ from the sketch $\mathcal{S}_1$.
We describe the merge operation on compactors in \Cref{the_compactor}. After the merge, we go through the compactors bottom-up and perform a compaction on each full compactor.

The update operation takes the next item from the input stream of the algorithm and adds it to the buffer of the level-0 compactor. If the compactor becomes full, this triggers a compaction on level $0$ and possibly also on higher levels (triggered by the addition of items from the previous compaction). Note that the update operation can be viewed as a merge with a trivial sketch containing one item.

\subsubsection{Answering queries}
For the rank query $y$, the sketch returns the value $\estrank(y)$, which is defined as
$$
\estrank(y) \eqdef \sum_{h = 0}^{H-1} 2^h \rank(y, B_h).
$$
To implement queries more efficiently, we can sort the items from all the compactors together with their weights and precompute weighted prefix sums. Then we answer a rank query by a single binary search over the sorted items (and one lookup to the prefix sums) and a quantile query by a single binary search over the prefix sums (and one lookup to the items).

\subsection{Adaptive compactor}\label{the_compactor}
The adaptive compactor consists of the \emph{buffer} $B$ of items and a stack $M$ of \emph{markers}, where each marker is represented by a copy of some item previously removed from $B$. 
As each marker is just a copy of an item, we can naturally compare markers and items (thus we can say that ``marker $m$ is smaller than item $y$''). 

For the analysis, we assume that both $B$ and $M$ are always sorted\footnote{In practice, we sort $B$ and $M$ before every compaction.} with the smallest item at position $B[0]$ and the smallest marker at position $M[0]$. 
Let us denote the smallest marker in $M_h$ by $\min(M_h)$ and if $M_h$ is empty or $h\geq H$ we define $\min(M_h)$ as $+\infty$.

The input stream $I$ of a compactor is the set of all items ever inserted into the buffer $B$, and the output stream $O$ is the set of all items ever moved from $B$ to the next-level compactor. In particular, the input stream of the level-$0$ compactor is the input of the algorithm, and the input stream of a level-$h$ compactor for $h>0$ is the output stream of the level-$(h-1)$ compactor.

The buffer $B$ is divided into \emph{sections} of length $K$, where each section is a range $B[iK, (i+1)K)$ for some integer $i \geq 0$. We denote the $i$-th section by $\se[i] = B[iK, (i+1)K)$ and analogously for ranges we have $\se[i, j] = B[iK, (j+1)K)$. A section is \emph{full} if it contains $K$ items. 

The space allocation strategy (see \Cref{strategy}) gives each compactor a \emph{section count} $S$. We stress that the actual number of sections of the compactor is not bounded by $S$ or any other fixed value, but the first $S$ sections have special meaning in the compaction algorithm.
The capacity of the compactor is given as 
\begin{equation*}
C \eqdef 2KS.
\end{equation*}
Thus, the compactor is \emph{full} when the first $2S$ sections are full.

We often use a subscript to denote the level of the compactor; thus, the compactor on level $h$ has buffer $B_h$, stack of markers $M_h$, section count $S_h$, capacity $C_h$, input stream $I_h$, and output stream $O_h$.

\begin{rem}
We set $C = 2KS$ to obtain better amortized update time. By setting $C = K\cdot(S+1)$, we get a worse update time, but the algorithm stays correct, and the space usage decreases almost twofold, which can be significant in practice. 
One can also pick another factor $f>1$ and set $C = K\cdot \max(S+1, \lceil fS \rceil)$ to get a trade-off between time and space usage; currently, we have $f=2$.
\end{rem}

\subsubsection{Merge}

Merging two compactors is straightforward: We just merge the sorted lists $B$ and $M$. The input and output streams are naturally the union of the original input and output streams.

\subsubsection{Marking}
We use the markers from the definition of an adaptive compactor to mark the sections of $B$. 

\begin{defn}[Marking]~\label{def:marking}
\begin{itemize}
\item \emph{Marking} is an injective function $\mar$ from markers to full sections of $B$ such that 
	$\forall m \in M, \forall y \in \mar(m): y < m$.

\item It is \emph{well-ordered} if for any two markers $m < m'$ we have $\mar(m) < \mar(m')$.

\item We say that the sections in the image of $\mar$ and their items are \emph{marked}; the rest of the sections and items are \emph{unmarked}.

\item For a marker $m$ that marks section $\se[i]$ in a given marking $\mar$, let its rank be defined as 
$\rank_{\mar}(m) = i+1$ (thus it is the number of sections with smaller items, including $\se[i]$). We usually drop the subscript if clear from context.
\end{itemize}

\end{defn}
We will ensure that such a marking always exists for any adaptive compactor (see \Cref{moving_markers}).
We will primarily use a canonical marking.

\begin{defn}[Canonical marking]
	For a sorted buffer $B$ and a sorted stack of markers $M$, the canonical marking is obtained by the following greedy procedure: Take the largest marker $m$ from $M$, find the rightmost full section $\se[i]$ such that $\se[i]$ is unmarked (so far) and all items in $\se[i]$ are smaller than $m$, and then set $\mar(m) = \se[i]$. Repeat while unused markers remain in $M$.
\end{defn}

\begin{obs}[Maximality of a canonical marking]\label{marking_properties}
	The canonical marking is the only well-ordered marking such that each marker has the maximum possible rank over all well-ordered markings.
\end{obs}
In the algorithm and in most of the analysis, we use only the canonical marking. Thus, for simplicity, we will only use the term \emph{marking} instead of the \emph{canonical marking}, and we explicitly state if we encounter a marking that may not be canonical.

\subsubsection{Compaction}
For the compaction, we will need a parameter $G$, the \emph{gap}, which we define in \Cref{space_allocation}.

The compaction operation is performed on a single level-$h$ compactor when it is full, i.e., $\abs{B} \geq C$ (thus, sections $\se[0, 2S-1]$ are full). If the number of items in $B$ is odd, we temporarily remove the largest item from $B$, and we add it back to $B$ after the compaction is finished.

First, we find the (canonical) marking of $B$ by $M$. Let $i < S$ be the largest index such that $\se[i]$ is unmarked (we prove the existence of such a section in \Cref{obs:unmarked_head}). The items from sections $\se(i,\infty]$ participate in the compaction (and there are at least $SK$ of them). We remove from $M$ all the markers that mark some of the sections $\se(i,\infty]$.

Next, we determine whether the compaction is \emph{marking} or not.
Let $y$ be the smallest of the participating items,
and recall that $\min(M_\ell)$ is the smallest marker in $M_\ell$ or $+\infty$ if $M_\ell = \emptyset$ or $\ell\ge H$.
If $y \leq \min(M_{h+G})$, we create a new marker $m = y$ (which is a copy of $y$) and push it to $M_h$. We say that this compaction is \emph{marking}. Otherwise, we do not create a new marker, and the compaction is \emph{non-marking}.

The participating items are removed from $B$.
We also toss a fair coin, and with probability $1/2$, the odd-indexed removed items are added to the next-level buffer; we say that those items are \emph{promoted}. In the other case, the even-indexed items are promoted. The remaining participating items are evicted from memory.

\subsection{The space allocation strategy}\label{space_allocation}

The space allocation strategy is a function that, for any state of all the compactors (present items and markers), gives the section count $S$ for each compactor. We explain it in two stages. First, for any compactor, we define a non-negative real counter $T_h$ which is a function of $B_h$ and $M_h$. Then, we define the integer section count $S_h$ as a function of all the counters $T_{h'}$ for $h' \geq h$.

\subsubsection{The counters}\label{counters}
Let us define the counter $T_h$ for a given level $h < H$. We define it with respect to a marking, and we use the canonical marking unless stated otherwise.

Given $B_h$ and its marking (by $M_h$), we define an integer weight of each section. Given the weights, we define a real value of each section. We use these values to define a non-negative real value of each suffix of the sections, and this finally gives us the value of the counter $T_h$. The process is depicted in \Cref{fig:ctr}.

\begin{figure}
	\includegraphics[width=\textwidth]{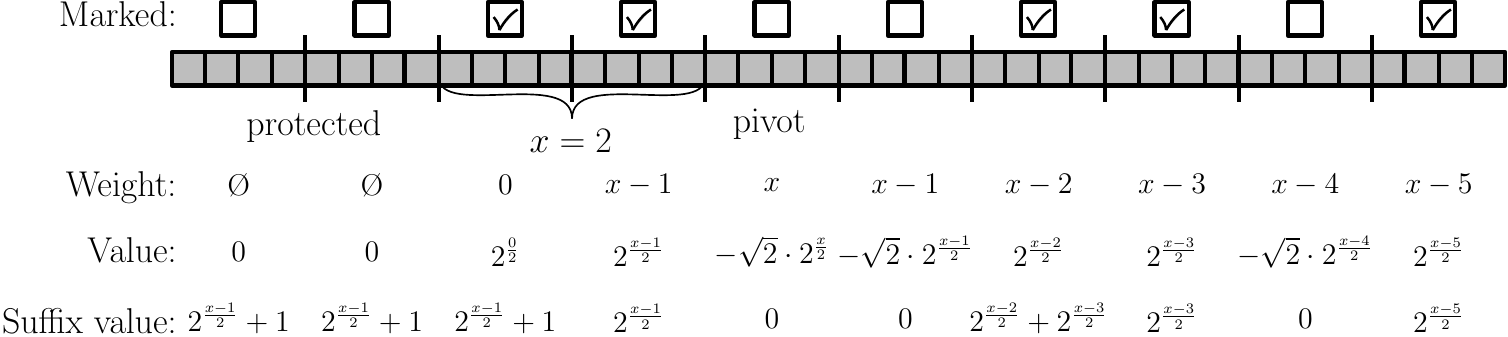}
	\caption{An example of a derivation of the value of $T$ from the marking. We have $x=2$ and $T = \sqrt{2}+1$ in the picture.}\label{fig:ctr}
\end{figure}

Let us call the leftmost two sections $\se[0, 1]$ (with the smallest items) the \emph{protected sections}. These sections have no weight, and we later prove that they are never marked (see \Cref{obs:unmarked_head}).
Let the \emph{pivot section} be the leftmost unmarked section which is not protected, and let $x \geq 0$ be the number of marked sections between the protected sections and the pivot (thus the pivot section is $\se[x+2]$). 

Let the pivot section have weight $w(x+2) = x$, 
let the $x$ marked sections between the protected sections and the pivot have weights from $0$ to $x-1$, increasing by one from the left to the right, and
let the rest of the sections to the right of the pivot have weights $x-1$ to $-\infty$, decreasing by one from the left to the right. More formally, we have $\forall i \leq x+2: w(i) = i-2$ and $\forall i \geq x+2: w(i) = 2x+2-i$.

For a section $\se[i]$ of weight $w(i)$ we define its value $t(i) \eqdef 2^{w(i)/2}$ if it is marked and $t(i) \eqdef -\sqrt{2}\cdot 2^{w(i)/2}$ if it is unmarked. We define the value of the protected sections as 0.

For any suffix $s$ of the sections containing no full section, we define its value as $T(s) = 0$, and
for any suffix $\se[i, \infty)$ where $\se[i]$ is full, we have $T[i, \infty) \eqdef \max(0, T[i+1, \infty) + t(i))$. The final counter $T_h$ of the level-$h$ compactor is naturally the value of the maximal suffix $T_h = T[0, \infty)$. 

Having defined $T_h$ for all non-imaginary levels, let us finally set $T_{\infty} = 2$.

\subsubsection{The strategy}\label{strategy}
Let $S_{\max}$ be the maximum value of $S$ ever achieved by any compactor.
For the whole sketch, we define the \emph{gap} $G$. Note that the value of $G_0$ is fixed, and the value of $G$ can only increase over time.
\begin{align}
G_0 &\eqdef \lceil \log \varepsilon^{-1} + \log \ln \delta^{-1} \rceil + 10, \label{g0def}\\
G &\eqdef \max\left(G_0, \ceil{\frac{1}{2}\log S_{\max} + 3\log\log S_{\max}}\right). \label{gdef}
\end{align}

We say that level $h$ is \emph{bounding} if and only if $\forall h' > h: T_{h'} < T_h$. In particular, the imaginary level $\infty$ is always bounding. For any level $\ell < H$, we define its bounding level $b(\ell)$ as the lowest bounding level $\ell'$ such that $\ell' \geq \ell + G$ (note that $b(\ell)$ is always defined due to the existence of the bounding level $\infty$). Finally, the section count $S_h$ of a level-$h$ compactor is given as
\begin{equation}
S_h \eqdef 6 + \lceil 4\log(T_h) \rceil + 2\cdot \max(0, \lceil T_h - T_{b(h)} \rceil) \label{sdef}\,.
\end{equation}

\begin{rem}
Note that $G$ depends on $S_{\max}$, which depends on $G$ through the definition of bounding levels. Thus, an increase of $S_{\max}$ can trigger an increase of $G$, which can trigger another increase of $S_{\max}$, and so on. 
For now, it is not clear that this process terminates; however, in \Cref{space_bound} we prove that $S_{\max} < 16H+64$ (\Cref{cor:smax}), which gives an upper bound on $S_{\max}$.
\end{rem}

\begin{rem}
Using $\log x = \log_2(\max(2, x))$ and $\log\ln \delta^{-1} > 1$ (as $\delta \leq 2^{-3}$), we have that $13 \leq G_0 \leq G$ and $10 \leq S_h$.
\end{rem}

\begin{rem}
In the actual implementation, we let the user pick the parameters $G_0$ and $K$ instead of $\varepsilon$ and $\delta$.
\end{rem}

\subsection{Basic observations}
Before we dive into the analysis, we need to state a few useful observations. 
In particular, we must prove that the algorithm is correct in the sense that the compaction operation as described in \Cref{the_compactor} can always be performed.
Even though it takes some work to prove the observations formally, we believe the statements themselves are mostly intuitive.

\Cref{moving_markers} captures a basic property of adaptive compactors: Under all operations, the markers either stay in place or move to the right. 

\begin{obs}[Markers only move to the right]\label{moving_markers}
The canonical marking of $B$ by $M$ always exists, and 
after any compaction operation, merge operation, or addition of new items to $B$,
the rank of each marker is at least its rank before.
\end{obs}

\begin{proof}
When $M$ is empty, the (empty) canonical marking surely exists. Let us show that no operation breaks this invariant, while the ranks of the markers do not decrease.

First, let us consider the compaction operation. During the compaction, we remove items from some suffix of the sections together with markers that mark those sections. Thus, if the canonical marking existed before, it did not cease to exist, and all the remaining markers mark exactly the same sections as before. 
Moreover, after the compaction, the rightmost full section is unmarked (by definition), and if the compaction is marking, we create a new marker that marks this previously unmarked section. This does not change the rest of the marking.

Second, let us consider the merge of compactors $Q_1, Q_2$ with parameters $B_1, M_1$ and $B_2, M_2$ respectively, resulting in a compactor $Q$ with parameters $B, M$. We use induction over $\abs{M}$. For $\abs{M}=0$, we have an empty marking and the statement holds.

For the induction step, let $m$ be the smallest marker in $M$, and without loss of generality, let $m$ be originally from $M_1$ and let $r$ be the rank of $m$ in $Q_1$. 
We modify $Q_1$ by removing $m$ from $M_1$ and the smallest $rK$ items from $B_1$. This decreases the rank of each marker in the canonical marking of $Q_1$ exactly by $r$, and the canonical marking still exists. Let $Q', B', M'$ be the result of the merge of this modified $Q_1$ with $Q_2$, and let $\mathcal{C}'$ be the canonical marking of $Q'$. By induction, $\mathcal{C}'$ exists, and the rank of each marker is at least its rank in its original buffer $B_2$ or in the modified $B_1$. 
Let us add the removed $rK$ items to $B'$, obtaining $B$. This increases the rank of each marker of $M'$ by $r$ (as their rank is in the number of sections), and thus it is at least its rank in the original buffer $B_1$ or $B_2$. Moreover, we add $m$ to $M'$ (obtaining $M$), and we let it mark the unmarked section $\se[r-1]$, obtaining some (not necessarily canonical) well-ordered marking where the rank of each marker is at least its rank before the merge. By \Cref{marking_properties}, the canonical marking with the same or higher rank of each marker exists.

Finally, the addition of new items to $B$ is equivalent to a merge with a compactor with no markers, and this situation is already covered by the argument above.
\end{proof}

Next, \Cref{ctr_range} gives us basic properties of the counter. Essentially, it shows that 
the possible ranges of the counter $T$ for different values of $x$ are disjoint and increasing in $x$.

\begin{obs}\label{ctr_range}
Let $x$ denote the number of marked sections between the protected sections and the pivot. We have
$$
2^{\frac{x}{2}}(1+\sqrt{2}) - (1 + \sqrt{2}) \leq T < 2^{\frac{x+1}{2}}(1+\sqrt{2}) - (1 + \sqrt{2}).
$$
\end{obs}

\begin{proof}
	This is by the definition of the counter. The marked sections to the left of the pivot have a cumulative value of $\sum_{i=0}^{x-1} 2^{i/2} = 2^{x/2}(1+\sqrt{2})-(1+\sqrt{2})$. The minimum value of $T$ is achieved when all the other sections are unmarked and contribute zero value. The maximum value of $T$ is achieved when all the sections to the right of the pivot are marked. Those sections contribute at most $\sum_{i=-\infty}^{x-1}2^{i/2} = 2^{x/2}(1+\sqrt{2})$ value, the pivot contributes $-2^{(x+1)/2}$ value and the marked sections to the left of the pivot again contribute $2^{x/2}(1+\sqrt{2})-(1+\sqrt{2})$. Summing these three values gives us the statement of the observation. 
\end{proof}

The following corollary is due to our definition of the logarithm $\log(x) \eqdef \log_2(\max(2, x))$.

\begin{cor}\label{xbound}
	$x \leq 2\log T$
\end{cor}
\Cref{obs:unmarked_head} shows the correctness of the compaction algorithm and justifies the definition of protected sections.

\begin{obs}\label{obs:unmarked_head}
	The compaction operation can always be performed, and the protected sections are never marked.
\end{obs}

\begin{proof}
	Recall that by definition we have $S \geq 4\log T + 6$ (see~\eqref{sdef}). Together with \Cref{xbound}, this gives us that $x+2 < S$. As the pivot is $\se[x+2]$, the compaction always stops to the right of the pivot and marks the pivot or some section with a higher index. Thus, the compaction can be performed, the protected sections cannot be marked by the compaction, and they stay unmarked by \Cref{moving_markers}.
\end{proof}

\Cref{obs:marking_compaction} may seem straightforward, but we state it as a standalone observation as it is critical in bounding the error of the sketch. In fact, proving this for relative compactors under full mergeability is very challenging~\cite{ReqSketch}. Note that the canonical marking always exists (by \Cref{moving_markers}) and that the compaction can always be performed (by \Cref{obs:unmarked_head}).

\begin{obs}\label{obs:marking_compaction}
	For any level-$h$ marking compaction $P$ with the smallest participating item $\min(P)$, there exists a set $U_P$ of $K$ items from the input stream $I_h$ of the compactor such that all of them are smaller than $\min(P)$.
	Crucially, for a fixed level $h$, the sets $U_P$ are pairwise disjoint.
\end{obs}

\begin{proof}
	Compaction $P$ creates marker $m$. If the marker is later deleted, the unique $K$ items are the items marked by $m$ during its deletion. If it is never deleted, we take the $K$ items marked by $m$ at the end of the algorithm.
	The disjointness property for a fixed level $h$ follows from the properties of a marking (Definition~\ref{def:marking}).
\end{proof}

\section{Analysis}\label{analysis}

In this section, we give an analysis of the sketch under full mergeability. We prove the space bound in \Cref{space_bound} (\Cref{lem:space}), the error guarantee in \Cref{sec:error_bound} (\Cref{error_bound}), and the time complexity in \Cref{time} (\Cref{l:time}). These together give the main theorem.

\begin{makethm}{theorem}{thm:mainthm} 
For any parameters $0 < \delta \leq 2^{-3}$ and $0 < \varepsilon \leq 1$, there is a randomized, comparison-based, fully mergeable streaming algorithm that, when processing an input consisting of $N$ items from a totally ordered universe $\mathcal{U}$, produces a summary $\mathcal{S}$ satisfying the following property.
Given $\mathcal{S}$, for any $y \in \mathcal{U}$ one can derive an estimate $\estrank(y)$ of $\rank(y)$ such that
$$
\pr \biggl[ \abs{\estrank(y) - \rank(y)} > \varepsilon \rank(y) \biggr] < \delta\,,
$$
where the probability is over the internal randomness of the algorithm.
The size of $\mathcal{S}$ is
$$
\O\left(\varepsilon^{-1} \cdot \log(\varepsilon N)\cdot  \log \log \delta^{-1} \cdot (\log \log \delta^{-1} + \log \varepsilon^{-1} + \log \log(\varepsilon N))\right)
$$
memory words, and the cumulative running time of update and merge operations is
$$
\O\left( N \cdot ( \log\log\log\delta^{-1} + \log\log(\varepsilon N) + \log\varepsilon^{-1} )\right).
$$
\end{makethm}

\begin{rem}
A straightforward implementation exactly following the description in \Cref{description} does not achieve the stated time complexity. In \Cref{time}, we briefly explain how to implement the algorithm efficiently, while preserving all the other properties.
\end{rem}

\subsection{The space bound}\label{space_bound}
To bound the sketch size, we define potentials $\Phi$ and $\Psi$. The definition of the potentials, particularly potential $\Psi$, is quite involved (\Crefrange{r:create}{r:merge}), but it is the cornerstone of our analysis of the space bound, together with \Crefrange{i:behind_maker}{i:final} regarding the potentials.
With the definitions and invariants in place,
the individual proofs of the invariants are not very long or involved. The last invariant gives us $T_h \leq 4(H-h) + \log S_{\max}$ for each $h$, which implies $S_{\max} \in \O(H)$ (\Cref{cor:smax}) and suffices to prove the space bound.

We start by stating a few more observations regarding marking, then we give some general intuition behind the potentials $\Phi$ and $\Psi$, and we state the formal definition of the potentials together with a more detailed explanation of their individual parts. After that, we formulate all the invariants, prove them one by one, and finish the analysis of the space bound.

\subsubsection{Observations about marking for the space bound}

Let us now state further observations regarding marking, which are useful in formulating the \Crefrange{r:create}{r:merge} and in proving \Crefrange{i:behind_maker}{i:final}.

\Cref{remove_smallest_marker} tells us that the smallest marker cannot be removed unless the compactor is in an extremely good state. It inspired \Cref{r:reset} in the choice of potential $\Psi$ and is used in the proof of \Cref{i:behind_maker}.

\begin{obs}\label{remove_smallest_marker}
	A compaction can remove the smallest marker from $M$ only if $T<1$ before the compaction.
\end{obs}

\begin{proof}
	A compaction can, by definition, remove only the markers marking sections to the right of the pivot. Thus, the smallest marker can be removed only if the number $x$ of marked sections to the left of the pivot is zero, and by \Cref{ctr_range} this is exactly when $T<1$.
\end{proof}

\Cref{obs:canonical_minimal} is only useful for proving \Cref{obs:tmerge,comp_increase}.

\begin{obs}\label{obs:canonical_minimal}
	Let $\mar_1, \mar_2$ be two well-ordered markings of $B$ by $M$ such that in both markings the protected sections are unmarked and $\forall m\in M: \rank_{\mar_1}(m) \geq \rank_{\mar_2}(m)$. Then, the value of $T$ under $\mar_1$ is at most the value of $T$ under $\mar_2$.
\end{obs}

\begin{proof}
	If the pivot is the same under both markings, the observation follows as the weights to the right of the pivot are decreasing with index. If the pivot under $\mar_1$ has a smaller index $x$ than under $\mar_2$, the observation follows from \Cref{ctr_range} (as the ranges for different values of $x$ are disjoint and increasing in $x$).
\end{proof}

\Cref{obs:tmerge} describes the behavior of our counter under the merge operation. It corresponds to \Cref{r:merge} and is used in the proof of \Cref{i:phibound}.

\begin{obs}\label{obs:tmerge}
	Let $Q$ be the result of a merge of two level-$h$ compactors $Q_1, Q_2$ with marker lists $M_1, M_2$ respectively, and without loss of generality let $\min(M_1) \leq \min(M_2)$.
	For the new counter $T$ after the merge, we have $T \leq T_1$, where $T_1$ is the counter of $Q_1$.
\end{obs}

\begin{proof}
	If at least one of the marker lists is empty, this is a direct consequence of \Cref{moving_markers,obs:canonical_minimal}. For the rest of the proof, let $\min(M_1) \leq \min(M_2)<+\infty$. Let $M = M_1 \cup M_2$.

	First, let us observe that if for some $i > 1$, sections $\se[i, i+1]$ are both unmarked, the value of the suffix $\se[i, \infty)$ equals zero. This is true as $T[i+2, \infty) < \sum_{j=i+2}^{\infty} 2^{w(j)/2} = 2^{w(i)/2} + 2^{w(i+1)/2}$ and the value of sections $\se[i, i+1]$ is $t(i) + t(i+1) = -\sqrt{2}\cdot (2^{w(i)/2} + 2^{w(i+1)/2})$. Let us call this fact the \emph{forgetting property}.
	
	Let us temporarily remove the smallest $2K$ items from the protected sections of both $B_1$ and $B_2$, and let us call these modified buffers $B_1', B_2'$. As the items were unmarked by \Cref{obs:unmarked_head}, both canonical markings still exist, and the rank of each marker decreased exactly by 2 compared to the original buffers. Let us merge the two modified compactors into $Q'$ with parameters $B', M$ and consider the canonical marking $\mathcal{C}'$ of $B'$ by $M$. By \Cref{moving_markers}, the ranks of the markers did not decrease. 
	After the merge, we add the $4K$ removed items to $B'$, obtaining $B$, and consider the canonical marking $\mathcal{C}$ of $B$ by $M$.
	
	Between $\mathcal{C}'$ and $\mathcal{C}$, the rank of each marker of $M$ that originated from $M_2$ increased exactly by 4. Of the markers originally from $M_1$, some of them increase their rank by 2, some by 3, and some by 4 compared to $\mathcal{C}'$. Let $\se[j]$ be the section of $B$ marked by $\min(M_2)$, and let us move all markers marking the sections to the left of $j$ (those are all originally from $M_1$) by up to two sections to the left in such a way that their rank is exactly two higher than in $\mathcal{C}'$. This gives us some (not necessarily canonical) well-ordered marking $\mar$ where the ranks of all markers are at least their ranks in $Q_1, Q_2$. 
	Moreover, as the rank of the marker marking $\se[j]$ is larger by 4 than in $\mathcal{C'}$
	and the rank of the markers marking $\se[0, j-1]$ is larger by 2 than in $\mathcal{C'}$, the sections $\se[j-2, j-1]$ are unmarked in $\mar$.
	
	By \Cref{obs:canonical_minimal}, the value of $T$ under $\mar$ is at least its value under the canonical marking, so it suffices to bound $T$ under $\mar$ by $T_1$. 
	As the sections $\se[j-2, j-1]$ are both unmarked
	and not protected, the value $T[j-2, \infty)$ equals 0 by the forgetting property.
	As the ranks of all markers are at least their ranks in their original buffers and all the markers marking sections to the left of $\se[j]$ originated from $M_1$, we have by \Cref{obs:canonical_minimal} that the value of $T$ under $\mar$ is at most $T_1$.
	
\end{proof}

\Cref{comp_increase} connects the increase of $T_h$ during a marking compaction to the value of $T_{b(h)}$. It is used to prove \Cref{i:psibound_local,i:psibound}.

\begin{obs}\label{comp_increase}
$T_h$ is increased only by a marking compaction, and in such a case, by at most
$$
\Delta T_h \leq 2^{- \max(0, T_h - T_{b(h)})}.
$$
\end{obs}

\begin{proof}
	The addition of new items and non-marking compactions do not increase $T_h$ by 
	\Cref{moving_markers,obs:canonical_minimal}, while for merge we have 
	\Cref{obs:tmerge}. It remains to analyze the case of a marking compaction.
	
	Let $\se[i]$ be the rightmost section not participating in the compaction (here, sections are only at level $h$). Before the compaction, $\se[i]$ was unmarked, and the compaction creates a marker $m$ that marks $\se[i]$. The rest of the full sections are marked exactly as before the compaction. 
		
	The increase of $T_h$ due to marking $\se[i]$ is at most $2^{w(i)/2}(1 + \sqrt{2})$ and the decrease of $T_h$ due to removing items and markers from $\se[i+1,S_h-1]$ is at least $2^{w(i)/2}(1+\sqrt{2}) - 2^{w(S_h-1)/2}(1+\sqrt{2})$.
	Thus the overall increase in $T_h$ is at most $2^{w(S_h-1)/2}(1+\sqrt{2})$.
	
	Using $w(S_h-1) = 2x + 2 - (S_h-1)$ by definition, $x \leq 2\log T_h$ by \Cref{xbound}, and  
	$S_h = 6 + \lceil 4\log T_h \rceil + 2\cdot \max\left(0, \ceil{T_h - T_{b(h)}}\right)$ by definition (see~\eqref{sdef}), we get
	$$
	\Delta T_h \leq 2^{w(S_h-1)/2}(1+\sqrt{2}) 
	\leq 2^{\frac{1}{2} (2x + 3 - S_h) + \frac{3}{2}}
	\leq 2^{2\log T_h + 3 - S_h/2}
	\leq 2^{-\max(0, T_h - T_{b(h)})}.
	$$
\end{proof}

Finally, \Cref{obs:monotonous_markers} tells us that the smallest markers are monotone. It is useful in proving \Cref{i:behind_maker,i:high_zero}.

\begin{obs}\label{obs:monotonous_markers}
	For each level $h$, $\min(M_h)$ is smaller than or equal to all items and markers on level $h+1$. In particular, this means that $\min(M_h) \leq \min(M_{h+1})$.
\end{obs}

\begin{proof}
	The proof is by a double induction, first over the levels top-down, then over the number of compactions on a given level $h$. 
	For $h \geq H-1$, the claim holds as there are no items on level $h+1$.
	
	For a given level $h<H-1$,
	the claim holds both before and after the first level-$h$ compaction, which forced us to create level $h+1$. 
	Whenever a new minimum item $y$ arrives at level $h+1$ due to a level-$h$ compaction $\mathcal{C}$, the item $y$ is smaller than $\min(M_{h+1})$ (by \Cref{obs:unmarked_head}),
	and by the induction hypothesis, it is at most $\min(M_{h+G})$ and thus $\mathcal{C}$ is marking. 
	The compaction $\mathcal{C}$ creates a marker which is at most $y$ but no smaller than $\min(M_h)$ as it is a marker in $M_h$ itself; it follows that $\min(M_h)$ is at most $y$.
	Whenever a compaction removes $\min(M_h)$, it is marking by the induction hypothesis, and it immediately creates an even smaller marker. The merge operation does not change the validity of the claim.
	The second part follows from the existence of a canonical marking (\Cref{moving_markers}).
\end{proof}

\subsubsection{The general intuition}\label{space_intuition}
Before we dive into the definitions, let us explain the intuition behind $\Phi$ and $\Psi$. Potential $\Phi$ is a sum of contributions $\Phi_h$ of all compactors, and the contribution of each compactor is exponential both in its counter $T$ and its level. Large $T$ means that the compactor is in a ``bad'' condition due to being subject to a lot of marking compactions, and it may require a large space; see the definition of $S_h$ in~\eqref{sdef}. Thus, $\Phi_h$ intuitively corresponds to the condition of the level-$h$ compactor. 

Whenever more $\Phi_h$ is created, it is due to a level-$h$ marking compaction that marked some $K$ small items. In this situation, we create an equal amount of potential $\Psi_{b(h)}$ and charge it to the small marked items (unless $b(h) = \infty$; see below). We maintain an invariant that the overall amount of $\Phi$ is at most the amount of $\Psi$.
The idea is that as those small items move up the levels of the sketch, they eventually reach $b(h)$ (which is the level in the worst condition above $h$), and as they are small, they move the markers on level $b(h)$ to the right and improve the condition of $b(h)$, decreasing $\Phi_{b(h)}$. When this happens, we delete the appropriate amount of $\Psi_{b(h)}$.
This is why $\Phi$ is exponential also in the level -- as the small items move up the sketch, their amount decreases exponentially, but the amount by which they can decrease $\Phi$ on a given compactor grows exponentially with the level.
Thus, for the increase of $\Phi_h$, we are eventually compensated by the decrease of $\Phi_{b(h)}$. 
The potential $\Psi$ that we assign to the items and markers using a number of carefully crafted rules is the medium through which this exchange happens, and we always have $\Psi_i \leq \frac{1}{4}\Phi_i$ for all $i<H$. 

The claims that $\Phi \leq \Psi$ and at the same time $\Psi_i \leq \frac{1}{4}\Phi_i\; \forall i<H$ may seem contradictory. 
However, there is a situation when we cannot get compensated by a decrease of $\Phi_{b(h)}$ for a marking compaction on level $h$. This is when all the compactors on higher levels are already in good condition, and this is the case where $b(h) = \infty$. In this case, we do not charge the new potential to any items, but we simply add it to a global counter $\Psi_\infty$.
Thus, all the potential $\Phi$ is asymptotically bounded by the amount of $\Psi_{\infty}$, since any $\Psi_h$ for $h<\infty$ is coupled with $\Phi_h$ and there is no $\Phi_{\infty}$. As $\Psi_\infty$ only grows during compactions, we can bound it by counting the compactions, which gives $\Psi_{\infty} \leq S_{\max}\cdot 2^{4H-4}$. This bounds $\Phi$ by $S_{\max}\cdot 2^{4H-4}$ up to a constant factor, which implies $T_h \leq 4(H-h)+\log S_{\max}$ for each $h$ and, in turn, $S_{\max} \in \O(H)$. With this knowledge in hand, we can finally bound the space using a telescoping property of the bounding levels.

\subsubsection{Definition of the potentials}
Let us formally define the potentials $\Phi$ and $\Psi$. The definition of $\Phi$ is straightforward.
For any level $h < H$, let 
\begin{equation}
\Phi_h = 2^{4h+T_h} - 2^{4h}\,. \label{eqn:def_Phi_h}
\end{equation}
In particular, this means that before the first compaction on level $h$, we have $\Phi_h = 0$. 
Let $\Phi = \sum_{h<H} \Phi_h$.

While $\Phi$ is just a function of the counters $T_h$, $\Psi$ depends on the history of the data structure in a more complex way. 
For any two levels $g<H$ and $h<H$, any item $y$ or marker $m$ present in the sketch is assigned a non-negative amount of potential $\Psi_{g, h}$. We denote this amount by $\Psi_{g, h}(y)$ or $\Psi_{g, h}(m)$.
Moreover, the sketch has a non-negative global counter $\Psi_\infty$, which is not assigned to any item or marker.
The counter $\Psi_\infty$ will never be decreased.
When two sketches are merged, the new value of $\Psi_\infty$ is the sum of the values of $\Psi_\infty$ in the two sketches.

For each $h<H$, $\Psi_{h}$ denotes the sum of $\Psi_{g, h}$ over all values of $g$ assigned to all items and markers in the sketch, and
$\Psi_{g, h}^i$ denotes the sum of $\Psi_{g, h}$ assigned to the items and markers on level $i$.
Finally, let $\Psi = \sum_{h<H} \Psi_h + \Psi_\infty$ be the total potential.

\begin{itshape}
We will maintain the convention of using the letter $g$ for the ``source'' level of the potential (i.e., at which level $g$ the increase of $\Phi_g$ created $\Psi_{g, h}^i$),
$i$ for the current level where the considered item is present, and
the letter $h$ for the ``target'' level, which is the bounding level $b(g)$ at the moment when $\Psi_{g,h}$ was created by \Cref{r:create}.
\end{itshape}

All the potential $\Psi$ (including the counter $\Psi_\infty$) is initially zero. Next, we define the rules for how it is created and moved around the sketch. Before stating the rules, we need a few more definitions.

By saying that all potential $\Psi$ from an item or marker $w$ \emph{moves} to an item or marker $w'$, or that $w'$ \emph{receives} the potential $\Psi$ of $w$ (or similarly), we mean that for each $g$ and $h$, we set $\Psi_{g, h}(w') \,+\!\!= \Psi_{g, h}(w)$ and $\Psi_{g, h}(w) = 0$.

For the purpose of analyzing the behavior of a particular level $h$ under a merge operation, let $\mathcal{S}^{(+h)}$ be the sketch where $\min(M_h)$ is larger and $\mathcal{S}^{(-h)}$ the sketch where $\min(M_h)$ is smaller (if they are equal, choose arbitrarily).
Further let us use superscripts $(+h)$ and $(-h)$ to denote parameters and potentials of the two sketches $\mathcal{S}^{(+h)}$ and $\mathcal{S}^{(-h)}$. Thus we have $\min(M_h^{(-h)}) \leq \min(M_h^{(+h)})$.

Let us finally state the rules. 
After each rule, we give a short intuitive explanation in \emph{italics}. 
This explanation is not part of the formal definition -- it is a simplified intuitive view, and like any simplification, it may be partially incorrect.

\begin{enumerate}[label = (R\arabic*)]
	\item\label[rule]{r:create}
	Let us denote by $\Delta \Phi_g \geq 0$ the amount by which $\Phi_g$ increased after a level-$g$ marking compaction. If $b(g) < \infty$, the new marker $m$ created by the compaction is charged with an amount $\Delta \Phi_g$ of potential $\Psi_{g, b(g)}$; formally, let $\Psi_{g, b(g)}(m) = \Delta\Phi_g$. If $b(g) = \infty$, we instead increase the global counter $\Psi_\infty$ by $\Delta\Phi_g$.

	\begin{itshape}
	This is the only way to create $\Psi$, and it is defined such that $\Phi$ is bounded by $\Psi$, as we later prove in \Cref{i:phibound}. If $b(g)=\infty$, there is no level above $g$ whose improvement could later compensate for the increase of $\Phi_g$, so we do not need to assign the potential to any items.
	\end{itshape}

	\item\label[rule]{r:deleted_marker}
	For each marker $m$ deleted during a compaction, all potential $\Psi$ assigned to $m$ is distributed evenly among the $K/2$ promoted items marked by $m$ during the compaction.
	
	\begin{itshape}
	This corresponds to the idea that each increase in $\Phi$ is charged to some small items -- through the marker and potential $\Psi$.
	\end{itshape}

	\item\label[rule]{r:deleted_item}
	During a marking compaction, each promoted item keeps all its $\Psi$ and each discarded item transfers all its $\Psi$ to its left promoted neighbor. If its left neighbor was not promoted, \Cref{r:last_item} applies instead.
	
	\begin{itshape}
	As the items move up the levels, the promoted items are charged by the potential of their neighbors in such a way that potential moves only to smaller items.
	Note that if the promoted items are marked, they also receive potential by \Cref{r:deleted_marker}.
	\end{itshape}

	\item\label[rule]{r:last_item}
	If the smallest item $t$ participating in a marking compaction was not promoted to the next level (i.e., it was evicted from memory), the marker $m$ created by the compaction receives all $\Psi$ that was assigned to $t$ (formally, for each $g$ and $h$ we set $\Psi_{g, h}(m) \,+\!\!= \Psi_{g, h}(t)$ and $\Psi_{g, h}(t) = 0$).
	
	\begin{itshape}
	If the smallest of the participating items was not promoted, we need some way to move its potential before the item is deleted; thus, we use the newly created marker, which has the same value (as it is a copy of the item).
	\end{itshape}

	\item\label[rule]{r:marker_change}
	For any level $i$, if the smallest marker in $M_{i+G}$ is about to change from $\min^{(\mathrm{old})}(M_{i+G})$ to $\min^{(\mathrm{new})}(M_{i+G})$ (either due to being removed by a compaction or due to a merge with another compactor with a smaller marker), then for each $g$ and each $h \geq i+G$, transfer all $\Psi_{g, h}$ from all items of $B_i$ and markers of $M_i$ to the marker $\min^{(\mathrm{old})}(M_{i+G})$. In the case of a merge, $B_i$ and $M_i$ here denote the buffer and the markers on level $i$ of the sketch $\mathcal{S}^{(+(i+G))}$, i.e., of the sketch whose smallest marker in $M_{i+G}$ changes, and this rule is applied after \Cref{r:merge}, for the levels $i$ in increasing order. (In the proof of \Cref{i:behind_maker}, we show that there is a nonzero $\Psi_{g,h}$ with $h \geq i+G$ on level $i$ only if $M_{i+G} \neq \emptyset$, so the marker $\min^{(\mathrm{old})}(M_{i+G})$ always exists when there is something to transfer.)
	
	\begin{itshape}
	This rule solves a certain problem with marking and non-marking compactions. An item $y$ can ``induce'' a marking compaction, but at some later time, another compaction in which $y$ participates is not marking due to a change in the minimum markers on some high levels. However, $y$ already caused a marking compaction which increased $\Phi$, and a corresponding amount of $\Psi$ was created. 
	
	To avoid such situations, in the case of a change of the smallest marker, we move all relevant $\Psi$ to the old smallest marker. This is possible since each marker is charged at most once and since $G$ is large enough so that the sum of $\Psi$ over a level $h$ is approximately the amount of $\Psi$ on one marker on level $h+G$ (as $\Phi$ is exponential in levels).
	\end{itshape}

	\item\label[rule]{r:factor}
	Consider any situation when $\Phi_h$ decreases from $\Phi_h^{\mathrm{old}}$ to $\Phi_h^{\mathrm{new}}$.
	In particular, in the case of a merge, for each level $h$, let $\Phi_h^{\mathrm{old}}$ denote the value $\Phi_h^{(-h)}$ and $\Phi_h^{\mathrm{new}}$ the value after a merge (thus we ignore $\Phi_h^{(+h)}$). 
	Let $f = \frac{2^{\max(2, T_h^{\mathrm{new}})}-1}{2^{\max(2, T_h^{\mathrm{old}})}-1}\leq 1$.
	For each $g$, for all items and markers in the sketch, multiply their assigned $\Psi_{g, h}$ by $f$.
	
	\begin{itshape}
	Conceptually, this rule says that if $\Phi_h$ decreases, we can decrease $\Psi_{g, h}$ by the same factor (and $\Phi$ is still bounded by $\Psi$). However, there are two caveats. First, if $T_h<2$, we do not do this, but we wait until $T_h<1$ and then use the decrease of $\Phi$ to eliminate all corresponding $\Psi_{g,h}$ by \Cref{r:reset}. Second, in the case of a merge we use $\Phi^{(+h)}$ to eliminate all $\Psi_{g, h}^{(+h)}$ by \Cref{r:merge}.
	\end{itshape}

	\item\label[rule]{r:reset}
	When $T_h < 1$, set $\Psi_{g, h}$ for all values of $g$ and for all items and markers in the sketch to 0.
	
	\begin{itshape}
	This corresponds to the case when some very small items reached level $h$ and moved the smallest marker. The compactor is in an excellent state, and there should be no $\Psi$ charged to it. We pay for this decrease of $\Psi_h$ by the decrease of $T_h$ (which corresponds to a decrease in $\Phi_h$) from 2 (when some $\Psi_{g, h}$ could have been created for the last time) to 1.
	\end{itshape}

	\item\label[rule]{r:merge}
	Before merging two sketches, for each level $h$, for all items and markers in $\mathcal{S}^{(+h)}$, we set their assigned $\Psi_{g, h}$ to 0 for all $g$. (This rule does not affect the counter $\Psi_\infty$.)

	\begin{itshape}
	By \Cref{obs:tmerge} we have $T_h \leq T_h^{(-h)}$; thus, all $\Phi_h^{(+h)}$ is released and we use it to eliminate all $\Psi_{g, h}^{(+h)}$.
	\end{itshape}
\end{enumerate}

\subsubsection{The invariants}
With the definition of potentials $\Phi$ and $\Psi$ in place, we state seven invariants regarding the potentials. All the invariants trivially hold for an empty sketch. Subsequently, we prove for each invariant that if it holds before a merge or compaction operation, it holds also after the operation (however, some of the invariants may not hold during the execution of the operation).
This usually requires going through the rules and checking that each of them preserves the invariant. Again, we give some intuition in italics, which is not part of the formal wording of the invariant.

\begin{enumerate}[label = (I\arabic*)]
	\item\label[inv]{i:behind_maker}
	Let $w$ be either a marker from $M_i$ or an item from $B_i$. If $\Psi_{g, h}(w) > 0$ for some $g$ and $h$, then we have 
	$w < \min(\min(M_h), \min(M_{i+G}))$.
	
	\begin{itshape}
	We have nonzero $\Psi$ only for small items, which can later decrease some counter on a higher level.
	\end{itshape}
	
	\item\label[inv]{i:non-marking} 
	All items $y$ and markers $m$ participating in a non-marking compaction have $\Psi_{g, h}(y) = \Psi_{g, h}(m) = 0$ for each $g$ and $h$.
	
	\begin{itshape}
	This is important as we do not create a marker and thus have no way to move potential $\Psi$ from the participating items.
	\end{itshape}
	
	\item\label[inv]{i:high_zero} 
	For any two levels $i > h$, for each $g$ we have $\Psi_{g, h}^i = 0$.
	
	\begin{itshape}
	The small items with nonzero $\Psi_{g, h}$ have no way to move above level $h$. If inserted, they ``reset'' $T_h$ to a value less than 1, and all $\Psi_{g, h}$ are eliminated by \Cref{r:reset}.
	\end{itshape}
	
	\item\label[inv]{i:psibound_local} 
	For any fixed levels $g, h$, let $p_{g,h}(i) = 2^{2(i + g) + \max(2,T_h) + 1}$.
	For any item $y \in B_i$ we have $\Psi_{g, h}(y) \leq p_{g,h}(i)/K$ and for any marker $m \in M_i$ we have $\Psi_{g, h}(m) \leq p_{g,h}(i)$.
	
	\begin{itshape}
	This is a key invariant. We bound $\Psi_{g, h}$ of any item or marker by $T_h$, allowing us to bound $\Psi_{h}$ by $\Phi_h$ in the next invariant.
	\end{itshape}
	
	\item\label[inv]{i:psibound} 
	For any level $h < H$ we have $\Psi_h \leq \frac{1}{4}\Phi_h$, and we have $\Psi_{\infty} \leq S_{\max}\cdot 2^{4H-4}$.
	
	\begin{itshape}
	This gives us the first half needed for the last invariant, bounding $\Psi$ by $\Phi$.
	\end{itshape}
	
	\item\label[inv]{i:phibound} 
	We have $\Phi \leq \Psi$.
	
	\begin{itshape}
	This gives us the second half of the proof -- that $\Phi$ is at most $\Psi$. This is probably easy to see just from the definition of the rules.
	\end{itshape}
	
	\item\label[inv]{i:final}
	For any level $h<H$, we always have $T_h \leq 4(H-h) + \log S_{\max}$.
	
	\begin{itshape}
	And this is the final claim. It implies that $T_h \leq 4H + \log H + 7$ for each $h$ (\Cref{cor:smax}), which allows us to bound the space.
	For the explanation of the importance of this, see \Cref{overview_elastic}.
	\end{itshape}
\end{enumerate}

In the proofs of the invariants, we use the following simple bound on the sizes of the buffers.

\begin{obs}\label{obs:buffer_size}
At any moment, including during the execution of an update or merge operation, every buffer $B_i$ contains fewer than $8KS_{\max}$ items and every $M_i$ contains fewer than $8S_{\max}$ markers.
\end{obs}

\begin{proof}
We say that level $i$ is \emph{processed} when the algorithm checks whether the level-$i$ compactor is full and, if so, performs a compaction on it. Whenever $B_i$ receives new items during an operation, level $i$ is processed later in the same operation (during an update, the levels are processed bottom-up as long as they receive new items, and during a merge, all levels are processed bottom-up). After that, $B_i$ contains fewer than $2KS_i \leq 2KS_{\max}$ items (for the value of $S_i$ at that time), and it does not receive any more items during the operation. Thus, at the beginning of an operation, every buffer contains fewer than $2KS_{\max}$ items. During an update, only $B_0$ receives an item from outside, and during a merge, every buffer receives fewer than $2KS_{\max}$ items from the other sketch. Apart from this, $B_i$ receives promoted items only during a compaction on level $i-1$, which promotes at most half of the items of $B_{i-1}$. As the compactions are performed bottom-up, this happens at most once during an operation and before the compaction on level $i$ (if any), after which $B_i$ receives no more items. Thus, by induction over $i$, $B_i$ always contains fewer than $4KS_{\max} + \frac{1}{2}\cdot 8KS_{\max} = 8KS_{\max}$ items. The bound on the number of markers follows, as every marker marks a different full section of $B_i$.
\end{proof}

\begin{proof}[Proof of \Cref{i:behind_maker}]
First let us recall that by \Cref{obs:monotonous_markers}, for any two levels $h < h'$ we have $\min(M_h) \leq \min(M_{h'})$. In particular, this means that $\min(\min(M_h), \min(M_{i+G})) = \min(M_{\min(h, i+G)})$ and that an increase in $G$ cannot violate the invariant.

When $\Psi_{g, h}$ is created on level $i$ by \Cref{r:create}, we have $h = b(i) \geq i+G$, and the statement holds by the definition of a marking compaction.

When moving $\Psi_{g, h}$ between items and markers by \Cref{r:last_item,r:deleted_item,r:deleted_marker}, the invariant is not broken as we move $\Psi_{g, h}$ only to smaller items and markers and only to higher levels.

Note that this invariant is not affected by \Cref{r:factor}.

When $\min(M_h)$ is about to stop being the smallest marker due to being deleted by a compaction, we have $T_h < 1$ by \Cref{remove_smallest_marker}, which implies $\Psi_{g, h} = 0$ for all $g$ by \Cref{r:reset}. When $\min(M_h)$ is about to stop being the smallest due to a merge, we have $\Psi_{g, h} = 0$ for all $g$ by \Cref{r:merge}.

Finally, suppose that $\min(M_{i+G})$ changes while some item or marker $w$ on level $i$ has $\Psi_{g,h}(w) > 0$. If $h < i+G$, the invariant follows from $w < \min(M_h) \leq \min(M_{i+G})$. If $h \geq i+G$, then $M_{i+G} \neq \emptyset$: level $h$ was bounding when $\Psi_{g,h}$ was created, so we had $T_h > T_\infty = 2$ and thus level $h$ had markers at that time; a level never loses all its markers (when a compaction removes $\min(M_h)$, it creates a new, smaller marker; see the proof of \Cref{obs:monotonous_markers}); and by \Cref{obs:monotonous_markers}, every level below a level with markers has markers too. If $h = i+G$, this is the case of $\min(M_h)$ changing, which we have already treated. If $h > i+G$, then $\min(M_{i+G})$ can change only by being removed by a compaction or due to a merge (creating a new marker cannot decrease $\min(M_{i+G})$, since by the definition of the canonical marking, each marker that is not removed by the compaction is at most the largest item of the rightmost non-participating section, and thus smaller than the new marker), and in both cases the invariant is restored by \Cref{r:marker_change}.
\end{proof}

\begin{proof}[Proof of \Cref{i:non-marking}]
It follows from \Cref{i:behind_maker} and the definition of a marking compaction.
\end{proof}

\begin{proof}[Proof of \Cref{i:high_zero}]
It follows from \Cref{i:behind_maker} and the fact that $\min(M_h)$ is smaller than or equal to all items and markers on level $h+1$ (\Cref{obs:monotonous_markers}).
\end{proof}

\begin{proof}[Proof of \Cref{i:psibound_local}]
First, a change of $T_h$ cannot break the invariant, as for any $g$ and any item or marker, the value of $p_{g, h}$ decreases at most by a factor of $f$ (due to the definition of $f$ and $\Phi$), while  its assigned $\Psi_{g, h}$ decreases exactly by the factor of $f$ by \Cref{r:factor}. Finally, \Cref{r:reset,r:merge} clearly cannot violate the invariant, as they just delete some $\Psi$.

It remains to show that \Cref{r:create,r:last_item,r:deleted_marker,r:deleted_item,r:marker_change} do not break the invariant. Let us prove it by induction over $i$. 
By \Cref{comp_increase}, a marking compaction increases $T_i$ by at most 
$$
\Delta T_i = 2^{-\max(0,T_i - T_{b(i)})} \leq 1.
$$ 
Using the fact that for $0<x\leq1$ we have $2^x-1\leq x$, we bound $\Delta\Phi_i$ using the definition of \Cref{r:create}:
$$
\Delta\Phi_i \leq 2^{4i+T_i}\cdot (2^{\Delta T_i}-1) \leq 2^{4i+T_i}\cdot \Delta T_i \leq 2^{4i+T_{b(i)}}\,,
$$
where $T_i$ is the original value of the counter; the last inequality uses $\Delta T_i \le 2^{-T_i + T_{b(i)}}$.
This gives us that by \Cref{r:create} the new marker acquires at most $p_{g,h}(i)/2$ of potential $\Psi_{g, h}$ for $h=b(i)$ (and no potential if $b(i)=\infty$), as we have $i = g$ at the time of applying \Cref{r:create}.
This suffices to prove the base case $i=0$, as markers on level 0 do not acquire any more $\Psi_{g, h}$ and all items on level 0 have $\Psi_{g, h} = 0$ for all $g$. Now let us assume the invariant holds for all $j<i$.

Any item on level $i$ has three sources of $\Psi_{g, h}$. It got at most 
$\frac{1}{K}p_{g,h}(i-1) \leq 2^{-2}\frac{1}{K}p_{g,h}(i)$ from the previous level, the same amount from its deleted neighbor (both by \Cref{r:deleted_item}), and $\frac{2}{K}p_{g,h}(i-1) \leq 2^{-1}\frac{1}{K}p_{g,h}(i)$ from some marker by \Cref{r:deleted_marker}. The sum is at most $p_{g,h}(i)/K$ and the claim holds by induction.

Any marker on level $i$ also has three sources of $\Psi_{g, h}$. 
First, by \Cref{r:create}, it gets at most $p_{g,h}(i)/2$, as we have $i = g$ while applying \Cref{r:create}.
Second, it gets at most $p_{g,h}(i) / K$ by \Cref{r:last_item}, and finally, by \Cref{r:marker_change} and \Cref{obs:buffer_size}, at most
\begin{align*}
\sum_{m \in M_{i-G}} \Psi_{g, h}(m) + \sum_{y \in B_{i-G}} \Psi_{g, h}(y) &\leq
8S_{\max}\cdot p_{g,h}(i-G) + 8S_{\max}K\frac{1}{K}\cdot p_{g,h}(i-G)\\
&= 16\, p_{g,h}(i-G)\cdot S_{\max} = 2^{4-2G}\cdot p_{g,h}(i) \cdot S_{\max}\\
&\leq p_{g,h}(i) \cdot 2^{-2}
\end{align*}
using $G \geq \frac{1}{2}\log S_{\max} + 3$ in the last inequality (see the definition in \eqref{gdef}). A marker receives potential by \Cref{r:marker_change} at most once, as it never becomes the smallest marker again after it stops being the smallest one. Again, the sum is at most $p_{g,h}(i)$ (as $K\geq 2^6$). This concludes the proof.
\end{proof}

\begin{proof}[Proof of \Cref{i:psibound}]
First, consider a level $h < H$.
Let us start by proving
\begin{equation}\label{phigh_bound}
\Psi_{g, h} \leq 2^{\max(2, T_h) + 4h + 8}/(h-g)^4.
\end{equation}

We claim that it suffices to prove~\eqref{phigh_bound} for the case when $h-g \geq G$. This is because $\Psi_{g, h}$ is created by \Cref{r:create} only if $h = b(g) \geq g+G$, and the other rules never change the source level $g$ or the target level $h$ of potential; thus, when $h-g < G$, no new $\Psi_{g, h}$ can be created. Moreover, when $T_h$ increases, the right-hand side of~\eqref{phigh_bound} does not decrease, and when $T_h$ decreases, \Cref{r:factor} multiplies $\Psi_{g,h}$ by a factor that is at most the factor by which the right-hand side of~\eqref{phigh_bound} decreases. Finally, \Cref{r:reset,r:merge} only decrease $\Psi_{g,h}$. Thus, if the bound holds for $h-g \geq G$, it also holds later when $h-g < G$ (if this ever happens).

For the proof of \eqref{phigh_bound}, we need that
\begin{equation}\label{eqn:log_S_max}
	\log(S_{\max}) \leq 2(h-g) - 4 \log(h-g) + 2.
\end{equation}
For this we use $h-g \geq G \geq \frac{1}{2}\log S_{\max} + 3 \log \log S_{\max} \geq 4$ (see the definition of $G$ in~\eqref{gdef}). As $h-g \geq 4$, the right side of~\eqref{eqn:log_S_max} grows with $h-g$ and thus it suffices to prove
\begin{align*}
\log(S_{\max}) \leq 2\cdot \left(\frac{1}{2}\log S_{\max} + 3 \log \log S_{\max}\right) &- 4\log\left(\frac{1}{2}\log S_{\max} + 3 \log \log S_{\max}\right)+2\,,
\end{align*}
or equivalently
\begin{align*}
\log\left(\frac{1}{2}\log S_{\max} + 3 \log \log S_{\max}\right)
&\leq \frac{3}{2}\log\log S_{\max}+\frac{1}{2}\\
\frac{1}{2}\log S_{\max} + 3 \log \log S_{\max}
&\leq \sqrt{2}\cdot \log^{\frac{3}{2}} S_{\max}\,,
\end{align*}
and this is true as $S_{\max} \geq 6$.

In the proof of~\eqref{phigh_bound} we further use \Cref{i:psibound_local} and the fact that by \Cref{i:high_zero} we have $\Psi_{g,h}^i = 0$ for $i > h$. By \Cref{obs:buffer_size}, every level contains fewer than $8KS_{\max}$ items and fewer than $8S_{\max}$ markers, and thus $\Psi^i_{g,h} \leq 16S_{\max}\,p_{g,h}(i)$ by \Cref{i:psibound_local}.
\begin{align*}
\Psi_{g, h} &\leq \sum_{i=0}^{h}\Psi_{g, h}^i \leq
\sum_{i=0}^{h} 16 S_{\max}\, p_{g,h}(i) =
16 S_{\max} \sum_{i=0}^{h} 2^{2(i+g)+\max(2, T_h) + 1}\\
&\leq 2^{\max(2, T_h) + 2g + 5} S_{\max} \sum_{i=0}^{h} 2^{2i} \\
&\leq 2^{\max(2, T_h) + 4h - 2(h - g) + 6 + \log S_{\max}}\\
&\leq 2^{\max(2, T_h) + 4h + 8} / (h-g)^4\,,
\end{align*}
where the last inequality uses~\eqref{eqn:log_S_max}.
Finally, we sum over all possible values of $g$ assuming $T_h \geq 1$ (otherwise we have $\Psi_h =0$ by \Cref{r:reset}). We use that every $g$ with $\Psi_{g,h} > 0$ satisfies $h-g \geq G_0 \geq 13$, as $\Psi_{g,h}$ is created only if $h = b(g) \geq g+G$ (again see \eqref{g0def}):
$$
\Psi_h \leq \sum_{g = 0}^{h-13} 2^{\max(2, T_h) + 4h + 8} / (h-g)^4 \leq
2^{\max(2, T_h) + 4h + 8} \sum_{i=13}^{\infty}i^{-4}
\leq 2^{\max(2, T_h) + 4h + 8 - 12}
\leq \frac{1}{4}\Phi_h\,,
$$
where the last inequality follows from $\Phi_h = 2^{4h}(2^{T_h} - 1) \geq 2^{4h+\max(2,T_h)-2}$ for $T_h \geq 1$.

Now, let us bound $\Psi_\infty$. By \Cref{r:create}, the counter $\Psi_\infty$ is increased only by marking compactions on levels $g$ with $b(g) = \infty$ (and in a merge, the counters of the two sketches are added), and it never decreases. For a marking compaction on level $g$ with $b(g) = \infty$, the computation in the proof of \Cref{i:psibound_local} gives $\Delta\Phi_g \leq 2^{4g+T_{\infty}} = 2^{4g+2}$.
Let $N_g$ be the number of compactions performed on level $g$ so far, including the compactions performed in all the sketches that were merged into the current sketch. We have $N_{H-1} = 0$, as there would be a level $H$ if level $H-1$ had performed a compaction.
Each compaction on level $g$ promotes at least $S_gK/2 \geq 5K$ items, and thus $N_g \leq \abs{I_{g+1}}/(5K)$.
Furthermore, every item of the input stream $I_j$ of a level-$j$ compactor is either still in $B_j$, or it participated in exactly one compaction, and every compaction promotes exactly half of the participating items (recall that the number of participating items is even). This gives us $\abs{I_j} = \abs{B_j} + 2\abs{I_{j+1}}$ for every $j<H$, where $I_H = \emptyset$. Using $\abs{B_j} < 8KS_{\max}$ (\Cref{obs:buffer_size}), we obtain
$$
\abs{I_{g+1}} = \sum_{j=g+1}^{H-1} 2^{j-g-1}\abs{B_j} < 8KS_{\max}\cdot 2^{H-g-1}
$$
and thus $N_g < \frac{4}{5}\,S_{\max}\cdot 2^{H-g}$. Putting this together, we get
$$
\Psi_\infty \leq \sum_{g=0}^{H-2} 2^{4g+2}N_g \leq \frac{16}{5}\,S_{\max}\, 2^{H}\sum_{g=0}^{H-2}8^g \leq \frac{16}{35}\,S_{\max}\,2^{H}\,8^{H-1} = \frac{2}{35}\,S_{\max}\cdot 2^{4H} \leq S_{\max}\cdot 2^{4H-4}\,.
$$
\end{proof}

\begin{proof}[Proof of \Cref{i:phibound}]
A non-marking compaction does not break the invariant as it does not increase $\Phi$ and does not decrease $\Psi$ by \Cref{i:non-marking}. A marking compaction does not break the invariant by \Cref{r:create}, as the created potential is either assigned to the new marker or added to the counter $\Psi_\infty$.

\Cref{r:last_item,r:deleted_marker,r:deleted_item,r:marker_change} do not change the value of $\Psi$, nor do the update and merge operations (in the sense of adding an item into the level-0 compactor or merging the compactors and adding up their counters $\Psi_\infty$, but without the triggered compaction operations).

In \Cref{r:factor} we multiply $\Psi_h$ by a factor $f$, whereas $\Phi_h$ is multiplied by some factor $f' \leq f$. As $\Psi_h \leq \frac{1}{4}\Phi_h$ by \Cref{i:psibound}, $\Phi_h$ is decreased by a larger amount than $\Psi_h$.

In \Cref{r:reset}, consider the decrease of $\Phi_h$ from the value $2^{4h+2} - 2^{4h}$ to the current value which is at most $2^{4h+1}-2^{4h}$. Due to the definition of \Cref{r:factor}, this decrease of $\Phi_h$ was not charged to any decrease of $\Psi_h$ yet. This decrease amounts to at least $2^{4h+1}$, while the current value of $\Psi_h$ is at most $2^{4h-2}$ by \Cref{i:psibound}, and thus we can charge the removal of $\Psi_h$ to the decrease of $\Phi_h$. Note that we never charge any potential twice, as $h$ needs $T_h \geq 2$ to become a bounding level again for some $\Psi_h$ to be created.

Finally, to analyze \Cref{r:merge},
we use the fact that \Cref{obs:tmerge} and the definition of $\Phi_h$ imply that the new potential $\Phi_h$ is at most $\Phi_h^{(-h)}$; thus, the overall decrease of $\Phi_h$ is at least $\Phi_h^{(+h)}$ and we can charge the removal of $\Psi_h^{(+h)}$ to this decrease as we have $\Psi_h^{(+h)} \leq \frac{1}{4} \Phi_h^{(+h)}$ by \Cref{i:psibound}. Note that we never charge any potential twice, as in \Cref{r:factor}, we ignore $\Phi_h^{(+h)}$. 
\end{proof}

\begin{proof}[Proof of \Cref{i:final}]
Using \Cref{i:psibound,i:phibound} and the definition of $\Phi$ we have
\begin{align*}
\Phi &\leq \Psi = \sum_{h<H} \Psi_h + \Psi_{\infty}
\leq\frac{1}{4}\sum_{h<H}\Phi_h + \Psi_{\infty} = \frac{1}{4}\Phi + \Psi_{\infty}\,,\\
\frac{3}{4}\Phi &\leq \Psi_{\infty} \leq S_{\max}\cdot 2^{4H-4}.
\end{align*}
For any level $h<H$, if $T_h < 1$, the invariant holds. Otherwise, we have $\Phi_h = 2^{4h}(2^{T_h}-1) \geq 2^{4h+T_h-1}$, and thus
\begin{align*}
2^{4h+T_h-1} &\leq \frac{4}{3}\cdot S_{\max}\cdot 2^{4H-4} < S_{\max}\cdot 2^{4H-3}\,,\\
T_h &< 4(H - h) + \log S_{\max} - 2\,.
\end{align*}
\end{proof}

The last invariant together with the definition of the section count gives us an upper bound on $S_{\max}$.

\begin{cor}\label{cor:smax}
We have $S_{\max} < 16H+64$, and thus $T_h \leq 4H + \log H + 7$ for every level $h$.
\end{cor}

\begin{proof}
By the definition of $S_h$ in~\eqref{sdef} and as $T_{b(h)} \geq 0$, we have
$S_h \leq 6 + (4\log T_h + 1) + 2(T_h + 1) = 2T_h + 4\log T_h + 9$ for every level $h$.
Consider the level $h$ and the moment at which the value $S_{\max}$ was attained. As $H$ and $S_{\max}$ never decrease, \Cref{i:final} applied at that moment gives $T_h \leq 4H + \log S_{\max}$, where $H$ and $S_{\max}$ denote the current values. Hence
$$
S_{\max} \leq 8H + 2\log S_{\max} + 4\log(4H + \log S_{\max}) + 9\,.
$$
The right-hand side is smaller than $S_{\max}$ whenever $S_{\max} \geq 16H+64$, and thus $S_{\max} < 16H+64$. Finally, \Cref{i:final} and $\log(16H+64) \leq \log H + 7$ give $T_h \leq 4H + \log H + 7$.
\end{proof}

\subsubsection{The proof of the space bound}
To finish the analysis, we apply \Cref{cor:smax} together with a telescoping argument regarding the section counts of bounding levels. For this we need another simple technical observation.

\begin{obs}\label{buffer_bound}
Between operations, $B_h$ always contains fewer than $2K(S_h+2)$ items.
\end{obs}

\begin{proof}
Right after level $h$ is processed (as defined in the proof of \Cref{obs:buffer_size}), $B_h$ contains fewer than $2KS_h$ items, but $S_h$ may decrease later, so $B_h$ may contain more than $2KS_h$ items between operations. Consider the last time when level $h$ was processed. Since then, $B_h$ and $T_h$ have not changed, and every level above $h$ has performed at most one compaction: this holds for the rest of the operation during which level $h$ was processed, and every later operation has been an update that performed no compaction above level $h$ (as a merge processes all levels and an update processes level $h$ before performing a compaction on a higher level). By \Cref{comp_increase}, a counter increases only by a compaction on its level and by at most 1, and thus every counter $T_{h'}$ with $h' \geq h+G$ has increased by at most 1. As $T_{b(h)} = \max(T_\infty, \max_{h+G \leq h' < H} T_{h'})$ (the highest level attaining this maximum is bounding, and no level between $h+G$ and this level is bounding) and as $G$ never decreases, $T_{b(h)}$ has increased by at most 1, and thus $S_h$ has decreased by at most 2. 
\end{proof}

\begin{lemma}[The space bound]\label{lem:space}
The size of the sketch is 
$$
\O(\varepsilon^{-1}\cdot \log(\varepsilon N)\cdot \log \log \delta^{-1} \cdot (\log \log(\varepsilon N) + \log \varepsilon^{-1} + \log \log \delta^{-1} ))
$$
memory words.
\end{lemma}

\begin{proof}
First, let us note that all the additional space required by the sketch is asymptotically dominated by the number of items stored in the buffers $B_i$ of the compactors. Thus, it suffices to bound the number of stored items. By \Cref{buffer_bound} this is at most $2K(S_h+2)$ at any level $h$, and as $S_h \geq 10$, the number of stored items is at most $2K\cdot \sum_{h=0}^{H-1}(S_h+2) \leq 3K\cdot \sum_{h=0}^{H-1}S_h$.

Let us divide each $S_h$ into two parts (to recall the definition, see~\eqref{sdef}). The first part has size $4\log T_h + 6 + 3$ and the second part has size at most $2\cdot \max(0, T_h - T_{b(h)})$ (we got rid of both of the ceiling functions by adding $+3$ to the first part). By \Cref{cor:smax}, we have $T_h \leq 4H + \log H + 7 \leq 12H$, and thus the first part has size at most $4\log(12H) + 9 \leq 4\log H + 24$. Hence, the sum of all the first parts is at most $H\cdot (4\log H + 24)$.

To bound the sum of the second parts, let us consider a different definition of the bounding level. Level $h$ is bounding if $T_h > 2$ and for all integers $i>0$ and levels $h' = h+iG$ we have $T_h > T_{h'}$. Further, let us define $b'(h)$ as the smallest bounding level $h'$ such that $h' = h + iG$ for some integer $i>0$, or $\infty$ if such a level does not exist. Observe that with these new definitions, the section count of each level is greater than or equal to the section count with the original definitions, since for each $h$, the levels taken into consideration as $b'(h)$ are a subset of the levels taken into consideration as $b(h)$ (and we pick the level with the largest $T$ out of them). Thus, it suffices to bound the sum of $S_h$'s for this new definition. 

For a level that is not bounding, the second part of $S$ is 0, and for a level $h$ that is bounding, it is at most $2\cdot(T_h - T_{b'(h)})$. 
For fixed $j \in [G]$, consider all bounding levels $h$ such that $h = j + iG$ for some integer $i$. The sum of all second parts for those bounding levels telescopes to 
$$
\sum_{\substack{h\equiv j \;(\mathrm{mod}\; G) \\ h\; \mathrm{is\; bounding}}} \hspace{-5mm}2 \cdot (T_h - T_{b'(h)}) \leq 2\cdot (T_{\mathrm{first}} - T_{\mathrm{last}}) \leq 8H + 2\log H + 14
$$
using \Cref{cor:smax}. Taking the sum over all $j \in [G]$, we can bound the sum of all second parts of $S_h$'s by $G\cdot(8H + 2\log H + 14)$.

Finally, we observe that $H \leq \log(\varepsilon N)$ as for any item on level $h$ we have deleted $2^h$ items of the original stream and we have at least $2/\varepsilon$ items on level $H-1$.

Next, we bound $G$. By \Cref{cor:smax}, we have $\log S_{\max} < \log(16H+64) \leq \log H + 7$, and thus
$$
\ceil{\frac{1}{2}\log S_{\max} + 3\log\log S_{\max}} \leq \frac{1}{2}(\log H + 7) + 3\log(\log H + 7) + 1 \leq 4\log H + 12\,.
$$
Together with $G_0 \leq \log \varepsilon^{-1} + \log \ln \delta^{-1} + 11$, the definition of $G$ in~\eqref{gdef} gives us $G \leq 4\log H + \log \varepsilon^{-1} + \log \ln \delta^{-1} + 12$.

To put all this together, we use $8H + 2\log H + 14 \leq 24H$ and $3K \leq 2^9\varepsilon^{-1}\log\ln\delta^{-1}$ (as $\ceil{x} \leq 2x$ for $x \geq 1$). The number of all stored items is at most
\begin{align*}
3K \sum_{h=0}^{H-1}S_h &\leq 3K\cdot\left(H\cdot(4\log H + 24) + G\cdot(8H + 2\log H + 14)\right) \\
&\leq 3K\cdot H\cdot (4\log H + 24 + 24G) \\
&\leq 2^9 \varepsilon^{-1} \log \ln \delta^{-1} \log(\varepsilon N) \cdot
(100\log \log (\varepsilon N) + 24\log \varepsilon^{-1} + 24 \log \ln \delta^{-1} + 312) \\
&\in \O\left(\varepsilon^{-1} \log \ln \delta^{-1} \log(\varepsilon N) \cdot
(\log \log (\varepsilon N) + \log \varepsilon^{-1} + \log \ln \delta^{-1})\right).
\end{align*}

\end{proof}

\subsection{The error bound}\label{sec:error_bound}

In this section, we prove the error guarantee of the sketch, considering the most general setting of full mergeability. However, we assume that the input does not depend on the internal randomness of the algorithm (thus our analysis does not cover adversarial inputs). We first give a few tools and definitions, and then we use them to prove the error bound (\Cref{error_bound}).

For the whole error analysis, let us fix an arbitrary item $y$,
and to simplify the notation, we write $\rank_{h}(y)$ instead of $\rank(y, I_h)$. The goal of this section is to bound the error of the sketch for a rank query $y$.

Recall that the error of an item $y$ is defined as $\err(y) \eqdef \estrank(y) - \rank(y)$ where $\rank(y)$ is the true rank of $y$ in the input stream and $\estrank(y) \eqdef \sum_{h=0}^{H-1}2^h \rank(y, B_h)$.
Thus, after a compaction on level $h$, there are two possibilities.
If the rank of $y$ among the compacted items is even, the error does not change, as, regardless of the random choice, half of the items smaller than or equal to $y$ get removed, and half are promoted to the next level.
On the other hand, if the rank of $y$ among the compacted items is odd, the error either increases or decreases by $2^h$ with equal probability.
The compactions that do not change the error are not important for the analysis, which leads us to the following definition.

\begin{defn}[Important items and compactions]\label{def:imp}~\\
\begin{itemize}
	\item \emph{Important items} are all items $y' \leq y$.
	\item \emph{Important compactions} are the compactions where the rank of $y$ among the compacted items is odd. Those are exactly the compactions affecting $\err(y)$.
	\item $t_h$ is the number of important compactions on level $h$. Note that it is a random variable. 
\end{itemize}
\end{defn}
Let $b_{h, i}$ denote the random $\pm 1$ variables generated during the $i$-th level-$h$ compaction, 
let $\mathcal{I}_h^{x}$ denote the subset of them generated during the first $\min(x, t_h)$ \emph{important} compactions on level $h$ and
let $\bin_h^x \eqdef \sum_{i\in \mathcal{I}_h^x}b_{h, i}$.
By the argument above, the error of the sketch can be expressed in the following way. 
\begin{obs}\label{obs:err}
$$\err(y) = \sum_h 2^h \bin_h^{t_h}$$
\end{obs}
This expression is much more suitable for the analysis than the original definition. To bound parts of the error separately, we also introduce a handy notation for the error on and up to a given level and for the error after the first $x$ important compactions.

\begin{defn}[The error]\label{defn:err}
\begin{align*}
\err_h(y) &\eqdef 2^h\bin_h^{t_h}\\
\err_{<h}(y) &\eqdef \sum_{\ell=0}^{h-1} \err_{\ell}(y)\\
\err_{h}^{x}(y) &\eqdef 2^h\bin_h^x
\end{align*}
\end{defn}
All random variables used in level-$h$ compactions are by definition i.i.d.\ Rademacher variables. This may not be generally true for any subset of those variables, if the inclusion/exclusion of a given variable is dependent on the realization of some of the variables (which is our case). 

However, whether a compaction is important is decided \emph{before} its coin is tossed. Hence, if $X_k$ denotes the coin of the $k$-th important level-$h$ compaction (or an independent fair coin if there is no such compaction), each $X_k$ is fair given $X_1, \dots, X_{k-1}$, so $X_1, X_2, \dots$ are i.i.d.\ Rademacher variables. As $\bin_h^x = \sum_{k=1}^{\min(x, t_h)} X_k$ and $\min(x, t_h) \leq x$, we can apply \Cref{fact} to $\err_h^x = 2^h \bin_h^x$.

\begin{fact}[Corollary of the Azuma--Hoeffding bound]\label{fact}
Let there be fixed constants $h$, $x\geq1$, and $\delta_h$, and let $X_1, X_2, \dots$ be a sequence of independent random variables such that each $X_i$ attains values $2^h, -2^h$ with equal probability. We have
$$
\pr \Biggl[\exists z\leq x: \abs{\sum_{i=1}^{z} X_i} \geq 2^h\sqrt{2x\ln(2\delta_h^{-1})} \Biggr] \leq \delta_h.
$$
\end{fact}
To utilize \Cref{fact}, we need a way to bound the number of important compactions on a given level. We start with a trivial (but useful) bound.

\begin{obs}[Trivial bound on $t_h$]\label{m_simple_bound}
Each important compaction removes at least one important item from $B_h$ as the rank of $y$ among the compacted items is odd. This gives us for each $h$
$$
t_h \leq \rank_{h}(y).
$$
\end{obs}
To bound $t_h$, we need a bound on $\rank_h(y)$.
Let us investigate the dependence of $\rank_{\ell}(y)$ on $\rank_{\ell-1}(y)$ (for $\ell>0$).
Letting $ O_{\ell-1}$ be the output stream of the compactor at level $\ell - 1$,
it holds that $\rank(y, I_\ell) = \rank(y, O_{\ell-1})$, so we are comparing the input and the output of the compactor on level $\ell - 1$. Some of the items from $I_{\ell - 1}$ can stay in the buffer $B_{\ell-1}$; the rest are subject to compaction. By \Cref{def:imp}, any non-important compaction always sends half of the participating important items to the output, and any important compaction always sends one important item more or less than half with equal probability. Neglecting the items staying in the buffer, we can bound the rank of $y$ in the input stream $I_\ell$ as
$$
\rank(y, I_{\ell}) = \rank(y, O_{\ell-1}) \leq \frac{1}{2}\left(\rank(y, I_{\ell -1}) + \bin_{\ell-1}^{t_{\ell -1}}\right).
$$
Using this inequality recursively yields \Cref{rank_behaviour}.

\begin{obs}[Bound on $\rank_h(y)$]\label{rank_behaviour}
For any level $h$, we have
$$
\rank_h(y) \leq 2^{-h}\rank(y) + \sum_{\ell=0}^{h-1}2^{-h+\ell}\bin_\ell^{t_{\ell}}
= 2^{-h}\cdot (\rank(y) + \err_{<h}(y)).
$$
\end{obs}
This is the best we get out of the trivial bound. To obtain something stronger (\Cref{mbound}), we need some additional definitions.
The \emph{final level} $F_y$ is the highest level that any important item has ever reached. As this is a random value, we also need its deterministic counterpart $E_y$. We later prove that we have $E_y \geq F_y$ with high probability. The levels $F_y, F_y', F_y''$ (or their deterministic versions) serve as pivots, dividing the sketch into 4 parts, each one analyzed separately in the proof.

\begin{defn}[Final level \& friends]\label{critical}
\begin{align*}
F_y  &\eqdef \text{the lowest level $h$ such that $\rank_h(y) \leq K$}\\
E_y  &\eqdef \text{the lowest level such that $2^{-E_y+3}\rank(y) \leq K$}\\
F_y' &\eqdef \max(0, F_y - \lceil \log\ln \delta^{-1} \rceil)\\
E_y' &\eqdef \max(0, E_y - \lceil \log\ln \delta^{-1} \rceil)\\
F_y''&\eqdef \max(0, F_y - G_0)\\
E_y''&\eqdef \max(0, E_y - G_0)
\end{align*}
\end{defn}
The following corollary is by the minimality of $E_y$.
\begin{cor}\label{minimalityE}
If $E_y > 0$, we have $2^{-E_y+4}\rank(y) > K$ and thus $2^{E_y} < 2^4\rank(y)/K$.
\end{cor}
As the smallest $K$ items are never compacted, there can be no error on level $F_y$ and above.
\begin{cor}\label{f_trivial}
For each $h\geq F_y$ we have $\err_h(y) = 0$.
\end{cor}
For levels above $F_y''$, we finally obtain the better bound on $t_h$, which is due to the definition of a marking compaction.
\begin{obs}[Better bound on $t_h$]\label{mbound}
For any level $h \geq F_y''$ we have
$$
t_h \leq \rank_h(y)/K.
$$
\end{obs}

\begin{proof}
We claim that for any $h \geq F_y''$, all level-$h$ important compactions are marking. If this were not the case, we would have some level $h' = h + G \geq h+G_0 \geq F_y$ such that at some time there was a marker $m$ smaller than $y$ on $h'$ (by the definition of a marking compaction), and this marker is a copy of an important item participating in the compaction that created $m$. This is a contradiction with the definition of $F_y$, as the smallest $K$ items never participate in a compaction.
Using this fact together with \Cref{obs:marking_compaction} (for each marking compaction there are $K$ unique small items) yields the desired bound.
\end{proof}

Finally, the following observation is needed for proving that $E_y \geq F_y$ (and thus $E_y' \geq F_y'$ and $E_y'' \geq F_y''$) with high probability.

\begin{obs}\label{large_frank}
If $F_y \geq G_0$, we have $\rank_{F_y''}(y) > 2^{G_0-2}K$; otherwise, we have $\rank_h(y) < 2^{-h+1}\rank(y)$ for all $h\leq F_y$.
\end{obs}

\begin{proof}
Analogously to \Cref{rank_behaviour} we give a deterministic upper bound on $\rank_h(y)$ for each $h > F_y''$ (so that all the important compactions are marking): 
$$\rank_h(y) \leq 2^{-1}(\rank_{h-1}(y)+\bin_{h-1}^{t_{h-1}}) \leq 2^{-1}(\rank_{h-1}(y)+t_{h-1}) \leq 2^{-1}(\rank_{h-1}(y)+\rank_{h-1}(y)/K).$$
Using this recursively for any level $F_y''<\ell\leq F_y$ we get
$$
\rank_\ell(y)\leq 2^{F_y''-\ell}\rank_{F_y''}(y) \cdot(1+1/K)^{\ell-F_y''} \leq
2^{F_y''-\ell}\rank_{F_y''}(y) \cdot \exp(G_0/K)< 2^{F_y''-\ell+1}\rank_{F_y''}(y),
$$
where in the last step we used $\exp(G_0/K) < 2$ by the definition of $G_0$ and $K$.

If $F_y \geq G_0$, we have $F_y = F_y'' + G_0$ and plugging in $\ell = F_y - 1 = F_y''+G_0-1$ gives 
$$
K < \rank_{F_y-1}(y) \leq 2^{-G_0+2}\rank_{F_y''}(y).
$$

On the other hand, if $F_y < G_0$, we have $F_y'' = 0$ and thus for each $\ell \leq F_y$ we get $\rank_\ell(y) < 2^{-\ell+1}\rank(y)$.
\end{proof}

With all this in place, we can finally dive into the lemma.

\begin{lemma}[The error bound]\label{error_bound}
With probability at least $1 - \delta$, we have
$$
\abs{\err(y)} \leq \varepsilon\rank(y).
$$
\end{lemma}

\begin{proof}
Let us first give a high-level overview of the proof. 
The main idea is that for each $h$ we bound the probability of $\abs{\err_h^{T_h}(y)}$ exceeding some small threshold $\varepsilon_h\rank(y)$, take the union bound over all $h$, and conditioning on the bound we prove for each $h$ that $t_h \leq T_h$ and thus $\err_h^{T_h} = \err_h$.
The lemma follows by the triangle inequality.
The proof consists of a few steps.
\begin{enumerate}[label=(S\arabic*)]
	\item In the first part, for each level $h$, we define $T_h$, $\varepsilon_h$ and $\delta_h$ such that $\sum \varepsilon_h \leq \varepsilon$ and $\sum \delta_h \leq \delta$.
	\item Second, we use \Cref{fact} for each $h<E_y'$ to prove that the probability of $\abs{\err_h^{T_h}(y)}$ exceeding $\varepsilon_h \rank(y)$ is at most $\delta_h$. Then we take the union bound.
	\item As a third step, we prove that conditioning on the union bound we have for each $h < E_y''$ that $t_h \leq T_h$, which implies $\err_h^{T_h}(y) = \err_h(y)$ and thus $\abs{\err_h(y)}<\varepsilon_h \rank(y)$.
	\item Then we use this fact to obtain $E_y \geq F_y$.
	\item This allows us to prove $\abs{\err_h(y)}<\varepsilon_h \rank(y)$ for $E_y'' \leq h < E_y'$ similarly to (S3).
	\item As a final step, we bound the error on levels $h\geq E_y'$.
\end{enumerate}

\subsubsection{(S1) Definition of $\varepsilon_h$ and $\delta_h$}\label{S1}
We define $T_h$, $\varepsilon_h$ and $\delta_h$ for each $h$ and prove that $\sum \varepsilon_h \leq \varepsilon$ and $\sum \delta_h \leq \delta$. 
\begin{equation*}
T_h \eqdef 
\begin{cases}
\floor{2^{-h+1}\rank(y)}, & h<E_y''\\[2mm]
\floor{2^{-h+1}\rank(y)/K}, & h \geq E_y''
\end{cases}
\end{equation*}

\begin{equation*}
\varepsilon_h \eqdef 
\begin{cases}
\varepsilon \cdot \sqrt{(E_y''-h)\cdot 2^{h-E_y''-9}}, & h<E_y''\\
\varepsilon \cdot \sqrt{(E_y'-h)\cdot 2^{h-E_y'-5}}, & E_y'' \leq h < E_y'\\
\varepsilon \cdot 2^{-4}\cdot\lceil \log\ln \delta^{-1} \rceil^{-1}, & E_y'\leq h < E_y\\
0, & E_y \leq h
\end{cases}
\end{equation*}
\begin{align*}
\begin{split}
\sum_{h=0}^{H-1} \varepsilon_h
&\leq \sum_{h=0}^{E_y''-1} \varepsilon_h + \sum_{h=E_y''}^{E_y'-1} \varepsilon_h + \sum_{h=E_y'}^{E_y-1} \varepsilon_h\\
&\leq \varepsilon \cdot \left( 
	\sum_{i=1}^{\infty}\sqrt{i\cdot 2^{-i-9}} + 
	\sum_{i=1}^{\infty}\sqrt{i\cdot 2^{-i-5}} + 
	\lceil \log\ln\delta^{-1} \rceil \cdot 2^{-4} \cdot \lceil \log\ln \delta^{-1} \rceil^{-1}
\right) \\
&\leq \varepsilon \cdot (0.19 + 0.74 + 0.0625) < \varepsilon
\end{split}
\end{align*}

\begin{equation*}
\delta_h \eqdef 
\begin{cases}
2\delta^{2(E_y''-h)}, & h<E_y''\\
2\delta^{2(E_y'-h)}, & E_y'' \leq h < E_y'\\
0, & E_y' \leq h
\end{cases}
\end{equation*}
For bounding $\sum \delta_h$ we use the fact that $\delta\leq2^{-3}$.
\begin{align*}
\begin{split}
\sum_{h=0}^{H-1} \delta_h
&\leq \sum_{h=0}^{E_y''-1} \delta_h + \sum_{h=E_y''}^{E_y'-1} \delta_h \leq 
	\sum_{h=0}^{E_y''-1} \delta\cdot 2^{-6(E_y''-h)+4} + \sum_{h=E_y''}^{E_y'-1} \delta\cdot 2^{-6(E_y'-h)+4}\\
&\leq \delta \cdot \left( 
	\sum_{i=1}^{\infty}2^{-6i+4} + 
	\sum_{i=1}^{\infty}2^{-6i+4} 
\right) < \delta
\end{split}
\end{align*}

\subsubsection{(S2) With probability at least $1-\delta$ we have $\abs{\err_h^{T_h}(y)}<\varepsilon_h\rank(y)$ simultaneously for all $h < E_y'$}\label{S2}
\Cref{fact} gives us 
\begin{equation*}
\pr \Biggl[\abs{\err_h^{T_h}(y)} \geq 2^h\sqrt{2T_h\ln(2\delta_h^{-1})} \Biggr] \leq \delta_h.
\end{equation*}

For the case $h < E_y''$ we assume $E_y'' = E_y - G_0$ (otherwise there is no such $h$). We plug in $\delta_h = 2\delta^{2(E_y''-h)}$ and $T_h \leq 2^{-h+1}\rank(y)$. With probability at least $1-\delta_h$ we obtain
$$
\abs{\err_h^{T_h}(y)} < 2^h\sqrt{2^{-h+2}\rank(y)2(E_y''-h)\ln(\delta^{-1})} = 
\sqrt{2^{h+3}\rank(y)\cdot(E_y''-h)\cdot\ln(\delta^{-1})\cdot 2^{-E_y''}\cdot 2^{E_y''}}.
$$

Using $E_y'' = E_y - G_0$ and the fact that $2^{E_y} < 2^4\rank(y)/K$ by \Cref{minimalityE} we get
$$
\abs{\err_h^{T_h}(y)} < 
\rank(y)\cdot\sqrt{2^{h+7}\cdot(E_y''-h)\cdot\ln(\delta^{-1})\cdot 2^{-E_y''}\cdot 2^{-G_0}/K} \leq
\varepsilon\rank(y)\cdot\sqrt{2^{h-E_y''-9}\cdot(E_y''-h)},
$$
where in the second step we used the definitions of $K$ and $G_0$.

The case $E_y'' \leq h < E_y'$ is analogous. Again we have $E_y' = E_y - \ceil{\log\ln\delta^{-1}}$ and we plug in $\delta_h = 2\delta^{2(E_y'-h)}$ and $T_h \leq 2^{-h+1}\rank(y)/K$, obtaining
\begin{align*}
\abs{\err_h^{T_h}(y)} &< 2^h\sqrt{2^{-h+2}\rank(y)2(E_y'-h)\ln(\delta^{-1})/K} \\
&= \sqrt{2^{h+3}\rank(y)\cdot(E_y'-h)\cdot\ln(\delta^{-1})\cdot 2^{-E_y'}\cdot 2^{E_y'}/K}
\end{align*}
with probability at least $1-\delta_h$.
We plug in $E_y' \leq E_y - \log\ln\delta^{-1}$ and $2^{E_y} < 2^4\rank(y)/K$ (by \Cref{minimalityE}), and in the last step we unroll $K$ from the definition.
\begin{align*}
\abs{\err_h^{T_h}(y)} &<
\rank(y)\cdot\sqrt{2^{h+7}\cdot(E_y'-h)\cdot\ln(\delta^{-1})\cdot 2^{-E_y'}\cdot 2^{-\log\ln\delta^{-1}}/K^2} \\
&\leq \varepsilon\rank(y)\cdot\sqrt{2^{h-E_y'-5}\cdot(E_y'-h)}.
\end{align*}
We conclude that for each $h < E_y'$ we have with probability at least $1 - \delta_h$
$$
\abs{\err_h^{T_h}(y)} < \varepsilon_h\rank(y).
$$

Let us denote by $\mathcal{E}$ the event that $\abs{\err_h^{T_h}(y)} < \varepsilon_h\rank(y)$ simultaneously for all $h < E_y'$. As we have $\sum \delta_h < \delta$ by step \hyperref[S1]{(S1)}, the union bound gives us that $\mathcal{E}$ happens with probability at least $1-\delta$. For the rest of the proof we condition on $\mathcal{E}$.

\subsubsection{(S3) Conditioning on $\mathcal{E}$ we have $\abs{\err_h(y)}<\varepsilon_h \rank(y)$ for each $h < E_y''$}\label{S3}
Proof by induction over $h$. 
For any $h$ we use \Cref{rank_behaviour}, which gives us $\rank_h(y) \leq 2^{-h}\cdot(\rank(y)+\err_{<h}(y))$. 
By the induction hypothesis, we have $\err_{<h}(y) \leq \varepsilon\rank(y) \leq \rank(y)$ (for $h=0$ this is trivial) and thus $\rank_h(y) \leq 2^{-h+1}\rank(y)$. 
The simple bound $t_h \leq \rank_h(y)$ from \Cref{m_simple_bound} gives us 
$$
t_h \leq \rank_h(y) \leq \floor{2^{-h+1}\rank(y)} = T_h.
$$ 
We conclude that $\err_h^{T_h}(y) = \err_h(y)$ and thus by conditioning on $\mathcal{E}$ from step \hyperref[S2]{(S2)} we have $\abs{\err_h(y)}<\varepsilon_h\rank(y)$.

\subsubsection{(S4) Conditioning on $\mathcal{E}$ we have $E_y \geq F_y$}\label{S4}
As $\rank_h(y)$ is a nonincreasing function of $h$, it suffices to prove that $\rank_{E_y''}(y) < \rank_{F_y''}(y)$. 
By \Cref{large_frank} we have either $\rank_{h}(y) < 2^{-h+1}\rank(y)$ for each $h \leq F_y$ or $\rank_{F_y''}(y) > 2^{G_0-2}K$.

In the first case, if we had $E_y < F_y$, the definition of $E_y$ gives us
$\rank_{E_y}(y)\le2^{-E_y+1}\rank(y)\le K/4<K$, which contradicts the proposition $E_y < F_y$ due to the definition of $F_y$.

In the second case, by step \hyperref[S3]{(S3)} we have $\abs{\err_{<E_y''}(y)} < \varepsilon\rank(y) \leq \rank(y)$, which together with \Cref{rank_behaviour} gives us 
$$
\rank_{E_y''}(y)\leq 2^{-E_y''}(\rank(y)+\err_{<E_y''}(y)) \leq 2^{-E_y''+1}\rank(y)
$$ 
and by the definition of $E_y$ this is at most $K\cdot2^{G_0-2}$. We conclude that $\rank_{E_y''}(y) < \rank_{F_y''}(y)$ and thus $E_y'' \geq F_y''$ and $E_y \geq F_y$.

Using this important fact, we have \Cref{f_trivial} and \Cref{mbound} also for levels $E_y$ and $E_y''$ instead of $F_y$ and $F_y''$. This allows us to continue bounding the error.

\subsubsection{(S5) Conditioning on $\mathcal{E}$ we have $\abs{\err_h(y)}<\varepsilon_h \rank(y)$ for each $h < E_y'$}\label{S5}
This step is analogous to step \hyperref[S3]{(S3)}. The main difference is using the stronger bound $t_h \leq \rank_h(y)/K$ from \Cref{mbound}, which is valid also for $E_y''$ by step \hyperref[S4]{(S4)} of the proof.

We proceed by induction over $h$. For $h<E_y''$ the statement was proved in step \hyperref[S3]{(S3)}.

For $h\geq E_y''$ we again use \Cref{rank_behaviour} together with the induction hypothesis to obtain $\rank_h(y) \leq 2^{-h}\cdot(\rank(y)+\err_{<h}(y)) \leq 2^{-h+1}\rank(y)$. 
The stronger bound on $t_h$ gives us 
$$
t_h \leq \floor{\rank_h(y)/K} \leq \floor{2^{-h+1}\rank(y)/K} = T_h.
$$ 
The final statement $\abs{\err_h(y)}<\varepsilon_h\rank(y)$ is again by conditioning on $\mathcal{E}$.

\subsubsection{(S6) Conditioning on $\mathcal{E}$ we have $\abs{\err_h(y)}<\varepsilon_h \rank(y)$ for each $h < E_y$}\label{S6}
For $h< E_y'$ the statement was proved in \hyperref[S5]{(S5)}. For $E_y'\leq h < E_y$ we again condition on $\mathcal{E}$ and proceed by induction using \Cref{mbound} ($t_h \leq \rank_h(y)/K$), which is valid also for $E_y''$ due to step \hyperref[S4]{(S4)} of the proof. 

In the induction step, by the same reasoning as in \hyperref[S3]{(S3)} and \hyperref[S5]{(S5)} we get
$$
\rank_h(y) \leq 2^{-h+1}\rank(y).
$$
Using a simple bound $|\bin_h^{t_h}| \leq t_h$ we obtain
$$
\abs{\err_h(y)} = 2^h\abs{\bin_h^{t_h}} \leq 2^h t_h \leq
2^h \rank_h(y) / K \leq 2^h\cdot 2^{-h+1} \rank(y) \cdot \varepsilon \cdot 2^{-5} \cdot 
\lceil \log \ln \delta^{-1} \rceil^{-1} = \varepsilon_h \rank(y).
$$

For $h \geq E_y$ we have $\err_h = 0$ by \Cref{f_trivial} and \hyperref[S4]{(S4)}. Using $\sum \varepsilon_h < \varepsilon$ by step \hyperref[S1]{(S1)}, the statement of the lemma follows from step \hyperref[S6]{(S6)} by the triangle inequality.
\end{proof}

\begin{rem}
	Having analyzed the error of our sketch, we comment on an issue in a similar error analysis of~\cite{ReqSketch};
	namely, the issue is in Section~4 of~\cite{ReqSketch} with a relatively simple streaming analysis and also in Section~6 of~\cite{ReqSketch} with a complex analysis of mergeability.
	There, the authors first prove that with high probability ranks decrease exponentially, i.e., $\rank_h(y) \leq 2^{-h+1}\rank(y)$, and then, \emph{conditioned on these rank bounds}, they subsequently provide a tight analysis of the error using a variant of the Hoeffding bound. However, the conditioning destroys the independence of the random bits used in the compactions, so Chernoff bounds are not directly applicable. (Another issue is that the Azuma--Hoeffding bound of \Cref{fact} should be used in~\cite{ReqSketch} as the number of random variables in the error sum is random.)

	Here, we circumvent the issue by proving all the probability bounds for fixed constants, taking a single union bound, and then proving that the constants are large enough.
	We expect that a similar approach would fix the proofs in~\cite{ReqSketch} without making their bounds worse up to constant factors.
\end{rem}

\subsection{The time complexity}\label{time}
Let us briefly sketch the implementation of the two operations of the sketch. 
There are some major differences from the description in \Cref{description} -- in an actual efficient implementation, we are lazy in sorting $B$ and $M$ and in calculating $T$.
In particular, the sorting happens only immediately before a level-$h$ compaction, and after a merge, we just take the value of $T_h$ from the sketch with the smaller $\min(M_h)$.
It is straightforward to adjust the analysis in Section~\ref{analysis} for this implementation of the algorithm.

\subsubsection{Update operation}
\begin{enumerate}[label = (U\arabic*)]
	\item\label{U1} Add $y$ to $B_0$ and immediately end if it is not full.
	\item\label{U2} Starting from level $i = 0$, while the current level $i$ is full (which is checked using the previously computed $S_i$), sort $B_i$ and $M_i$, find the marking, calculate $T_i$ and $S_i$, create the next level if it does not exist, do the compaction, calculate $T_i$ again, and move to the next level by incrementing $i$. We end at some level $h$ that is not full; there, we do not perform a compaction.
	\item\label{U3} Go through all the levels $i = h-1, \dots, 0$ in decreasing order, updating the bounding level $b(i)$ and recalculating $S_i$.
	\item\label{U4} If $G$ increased, go through all levels top-down and update bounding levels and $S_h$. If $G$ increased again, repeat until $G$ stops increasing.
\end{enumerate}

\subsubsection{Merge $\mathcal{S}$ with $\mathcal{S}'$}
\begin{enumerate}[label = (M\arabic*)]
	\setcounter{enumi}{-1}
	\item\label{M0} Without loss of generality, let $H \geq H'$ (the height of $\mathcal{S}$ is at least that of $\mathcal{S}'$). Set $G$ in $\mathcal{S}$ as the maximum of the gaps of the two sketches.
	
	\item\label{M1}
	For each level $h < H'$, copy the items from the smaller of the two buffers to the larger one and from the smaller stack $M$ to the larger one. 
	For the new level-$h$ compactor of $\mathcal{S}$, we use this resulting buffer and stack, and we set the value of $T_h$
	as the value of $T_h$ from the sketch where $\min(M_h)$ is smaller.
	\item\label{M2} Go through all the levels of $\mathcal{S}$ from $H'-1$ to the bottom and calculate bounding levels and section counts.
	\item\label{M3} Go through the levels of $\mathcal{S}$ bottom-up and for each full level $i$: 
	sort $B_i$ and $M_i$, find the marking, calculate $T_i$ and $S_i$, create the next level if it does not exist, do the compaction, and recalculate $T_i$ again. Stop when the current level $h$ is not full and $h > H'-1$.
	\item\label{M4} Go through all the levels $i$ of $\mathcal{S}$ from the stopping level $h$ to bottom level 0, and calculate bounding levels $b(i)$ and section counts $S_i$.
	\item\label{M5} If the gap $G$ increased (compared to the maximum of the two sketches before the merge), go through all levels of $\mathcal{S}$ top-down and update bounding levels and $S_h$. If $G$ increased again, repeat until $G$ stops increasing.
	\item\label{M6} Return $\mathcal{S}$.
\end{enumerate}

\begin{lemma}[Time complexity]\label{l:time}
The cumulative running time of the algorithm under any series of update and merge operations is at most
$$
\O\left(N \cdot (\log\log\log\delta^{-1} + \log\log(\varepsilon N) + \log\varepsilon^{-1} )\right),
$$
where $N$ is the number of items summarized.

\end{lemma}

\begin{proof}
	We charge all the time steps of the algorithm to the $N$ input items.
	
	\paragraph{Update operation lines \ref{U1} to \ref{U3}}
	During an update operation, we spend $\O(1)$ time for adding the item to level 0 by \ref{U1} and then $\O(\abs{B} \log \abs{B})$ time for each level where we performed a compaction by lines \ref{U2} and \ref{U3}. As after any compaction, we evict at least $\abs{B}/4$ items from memory, we can charge all this time to the evicted items, $\O(\log \abs{B})$ each, and we have $\O(\log \abs{B}) \subseteq \O(\log\log\log\delta^{-1} + \log\log(\varepsilon N) + \log\varepsilon^{-1})$ by the definition of $K$ and $S$ (see~\eqref{kdef} and~\eqref{sdef}) and by \Cref{cor:smax}.
	
	\paragraph{Merge operation lines \ref{M1} to \ref{M4}}
	During a merge operation, all the work done by lines \ref{M1} to \ref{M4} on levels where we performed a compaction (on line \ref{M3}) can be charged to the items evicted from memory during the compaction, as in the analysis of the update. 
	
	On any level $h$ without compaction, the time spent on lines \ref{M1} to \ref{M4} is linear in the smaller of the two buffers.
	Consider the first compaction on level $h$ in the future, and let $t$ be the time of this compaction. If the algorithm ends before such a compaction happens, let $t$ be the time of the end of the algorithm.
	
	All the states of the buffer $B$ after any compaction previously performed on level $h$ but before $t$ can be imagined as a binary tree where 
	leaves represent initial states (after a previous compaction or after the creation of the buffer),
	nodes with one child represent the addition of new items from the previous level, and 
	nodes with two children represent merge operations.
	Let the size of any node be the current number of items in $B$. 
	Thus, the size of a merge node is the sum of its two children.
	
	Let $r$ be the size of the root. The overall time of all the merges before $t$ is the sum over all the merge nodes of the size of their smaller child, and a simple induction shows that this is at most $r \log r$. Thus, we charge all the time spent on lines \ref{M1} to \ref{M4} to items removed from the buffer by the compaction at time $t$ or to the items remaining in the buffer at the end of the algorithm if $t$ is the end of the algorithm.
	
	\paragraph{Changes of $G$ by \ref{U4} and \ref{M5}}
	We prove that the cumulative time spent on lines \ref{U4} and \ref{M5} is at most $N \cdot (G - G_0)$.
	Any time we increase $G$, the bound increases by $N$, which covers the time needed to recalculate bounding levels and section counts. After a merge of two sketches summarizing $N_1$ and $N_2$ items, the bound holds as 
	$N_1 \cdot (G_1 - G_0) + N_2 \cdot (G_2 - G_0) \leq 
	(N_1 + N_2) \cdot (\max(G_1, G_2) - G_0)$.
	The overall time per item is at most $G - G_0 \in \O(\log\log(\varepsilon N))$ by the definition of $G$ in \eqref{gdef} and \Cref{cor:smax}.
	
	\paragraph{The rest of Merge}
	The time spent by each merge in \ref{M0} and \ref{M6} is constant, and there are at most $N$ merge operations. This concludes the proof.
\end{proof}

\section{Discussion and open problems}\label{conclusions}

\subsection{Discussion of our result}
We have developed adaptive compactors as a new building block for mergeable relative-error quantile sketches, replacing both relative compactors from~\cite{ReqSketch} and elastic compactors from~\cite{GribelyukSWY25}. We have also devised a space allocation strategy for our compactors and introduced the notion of marking and non-marking compactions. Using these tools, we obtained a relative-error sketch with three crucial properties:

First, it outperforms the state-of-the-art randomized comparison-based sketches in space consumption -- in particular, it exceeds the known lower bound of 
$\Omega(\varepsilon^{-1}\log (\varepsilon N) \log \log \delta^{-1})$ 
only by a factor of 
$\O(\log\log \delta^{-1} + \log \log (\varepsilon N) + \log \varepsilon^{-1})$. Moreover, any further improvement in the dependence on the parameter $\delta$ would yield new results for deterministic sketches, specifically improving upon the state-of-the-art bound of $\O(\varepsilon^{-1}\log^3 (\varepsilon N))$~\cite{zhang}.

Second, in contrast to~\cite{GribelyukSWY25}, our sketch is fully mergeable, allowing distributed quantile estimation or seamless parallelization of data processing.

Finally, its simple structure makes it suitable for practical usage, which we demonstrated by a proof-of-concept implementation\footnote{\url{https://github.com/domestomas/adaptive-sketch-experiments}}.

\subsection{Future research}
From the theoretical point of view, the most relevant question is finally closing the gap between our result and the known lower bound. It seems plausible that one could somehow get rid of the $\log\log (\varepsilon N)$ and $\log \varepsilon^{-1}$ factors. However, whether the additional $\log\log \delta^{-1}$ factor could be removed may be better approached via deterministic comparison-based algorithms, as discussed in \Cref{lower_bounds}.

The most important practical task is to develop a proper implementation, optimizing it using both standard tricks~\cite{kll_improvements} and new insights, and performing a detailed comparison to ReqSketch~\cite{ReqSketch} and $t$-digest~\cite{tdigest}, both in terms of space and time. If these experiments are favorable to our algorithm, it could replace ReqSketch~\cite{ReqSketch} in the standard open-source libraries, such as Apache DataSketches\footnote{\url{https://datasketches.apache.org}}.

\printbibliography

@article{GuptaSW-noncomparison_optimal,
	title={Optimal Quantile Estimation: Beyond the Comparison Model}, 
	author={Meghal Gupta and Mihir Singhal and Hongxun Wu},
	journal={2024 IEEE 65th Annual Symposium on Foundations of Computer Science (FOCS)},
	year={2024},
	pages={1137-1158},
	doi={10.1109/FOCS61266.2024.00075},
}

@article{kll_improvements,
	title={Streaming Quantiles Algorithms with Small Space and Update Time}, 
	author={Nikita Ivkin and Edo Liberty and Kevin Lang and Zohar Karnin and Vladimir {Bra\-ver\-man}},
	journal={Sensors},
	year={2022},
	eprint={1907.00236},
	archivePrefix={arXiv},
}

@ARTICLE{ReqSketch_improvements,
	title={Theory meets Practice at the Median: a worst case comparison of relative error quantile algorithms}, 
	author={Graham Cormode and Abhinav Mishra and Joseph Ross and Pavel Veselý},
	journal={KDD '21: Proceedings of the 27th ACM SIGKDD Conference on Knowledge Discovery \& Data Mining},
	year={2021},
	doi = {10.1145/3447548.3467152}
}

@article{ReqSketch,
	author       = {Graham Cormode and
	Zohar S. Karnin and
	Edo Liberty and
	Justin Thaler and
	Pavel Vesel{\'{y}}},
	title        = {Relative Error Streaming Quantiles},
	journal      = {J. {ACM}},
	volume       = {70},
	number       = {5},
	pages        = {30:1--30:48},
	year         = {2023},
	url          = {https://doi.org/10.1145/3617891},
	doi          = {10.1145/3617891},
	bibsource    = {dblp computer science bibliography, https://dblp.org}
	}

@ARTICLE{tdigest,
	title = {The t-digest: Efficient estimates of distributions},
	author = {Ted Dunning},
	journal = {Software Impacts},
	year = {2021},
	doi = {10.1016/j.simpa.2020.100049},
}

@article{add_lower,
	title={Tight Lower Bound for {Com\-pa\-ri\-son}-Based Quantile Summaries}, 
	author={Graham Cormode and Pavel Veselý},
	year={2020},
	journal={Proceedings of the 39th ACM SIGMOD-SIGACT-SIGAI Symposium on Principles of Database Systems},
	doi = {10.1145/3375395.3387650},
}

@article{KLL,
	title={Optimal quantile approximation in streams},
	author={Karnin, Zohar and Lang, Kevin and Liberty, Edo},
	journal={2016 {IEEE} 57th annual symposium on foundations of computer science ({FOCS})},
	year={2016},
	doi={10.1109/FOCS.2016.17}
}

@article{zhang,
	author = {Zhang, Qi and Wang, Wei},
	title = {An efficient algorithm for approximate biased quantile computation in data streams},
	year = {2007},
	doi = {10.1145/1321440.1321601},
	journal = {Proceedings of the Sixteenth ACM Conference on Information and Knowledge Management},
}

@article{q_digest_new,
	author = {Cormode, Graham and Korn, Flip and Muthukrishnan, S. and Srivastava, Divesh},
	title = {Space- and time-efficient deterministic algorithms for biased quantiles over data streams},
	journal = {Proceedings of the Twenty-Fifth ACM SIGMOD-SIGACT-SIGART Symposium on Principles of Database Systems},
	year = {2006},
	doi = {10.1145/1142351.1142389},
}

@article{q_digest,
  author = {Nisheeth Shrivastava and Chiranjeeb Buragohain and Divyakant Agrawal and Subhash Suri},
  title = {Medians and beyond: new aggregation techniques for sensor networks},
  journal = {Proceedings of the 2nd International Conference on Embedded Networked Sensor Systems},
  year         = {2004},
  doi          = {10.1145/1031495.1031524},
}

@article{GK,
	title = {Space-efficient online computation of quantile summaries},
	author = {Greenwald, Michael and Khanna, Sanjeev},
	year = {2001},
	journal = {SIGMOD Rec.},
	doi = {10.1145/376284.375670},
}

@article{MRL,
	author = {Manku, Gurmeet Singh and Rajagopalan, Sridhar and Lindsay, Bruce G.},
	title = {Random sampling techniques for space efficient online computation of order statistics of large datasets},
	year = {1999},
	journal = {SIGMOD Rec.},
	doi = {10.1145/304181.304204},
}

@misc{tene2015_latency_talk,
	title = {How {NOT} to Measure Latency},
	author = {Gil Tene},
	year = {2015},
	howpublished = {\url{https://www.youtube.com/watch?v=lJ8ydIuPFeU}},
	lastaccessed = {April 1, 2021}
	
}

@article{MassonRL19,
	Author = {Charles Masson and Jee E. Rim and Homin K. Lee},
	Bibsource = {dblp computer science bibliography, https://dblp.org},
	DOI={10.14778/3352063.3352135},
	Journal = {{PVLDB}},
	Number = {12},
	Pages = {2195--2205},
	Title = {DDSketch: {A} Fast and Fully-Mergeable Quantile Sketch with Relative-Error Guarantees},
	Volume = {12},
	Year = {2019}}

@article{agarwal2013mergeable,
	Author = {Agarwal, Pankaj K and Cormode, Graham and Huang, Zengfeng and Phillips, Jeff M and Wei, Zhewei and Yi, Ke},
	Journal = {ACM Transactions on Database Systems (TODS)},
	Number = {4},
	Pages = {26},
	Publisher = {ACM},
	Title = {Mergeable summaries},
	Volume = {38},
	Year = {2013},
	Doi = {10.1145/2500128}
	}

@article{GribelyukSWY25,
	author       = {Elena Gribelyuk and
	Pachara Sawettamalya and
	Hongxun Wu and
	Huacheng Yu},
	editor       = {Yossi Azar and
	Debmalya Panigrahi},
	title        = {Near-Optimal Relative Error Streaming Quantile Estimation via Elastic
	Compactors},
	journal    = {Proceedings of the 2025 Annual {ACM-SIAM} Symposium on Discrete Algorithms,
	{SODA} 2025, New Orleans, LA, USA, January 12-15, 2025},
	pages        = {3486--3529},
	publisher    = {{SIAM}},
	year         = {2025},
	url          = {https://doi.org/10.1137/1.9781611978322.115},
	doi          = {10.1137/1.9781611978322.115},
	bibsource    = {dblp computer science bibliography, https://dblp.org}
}

@article{GuptaZ03,
	author       = {Anupam Gupta and
	Francis Zane},
	title        = {Counting inversions in lists},
	journal    = {Proceedings of the Fourteenth Annual {ACM-SIAM} Symposium on Discrete
	Algorithms, January 12-14, 2003, Baltimore, Maryland, {USA}},
	pages        = {253--254},
	publisher    = {{ACM/SIAM}},
	year         = {2003},
	url          = {http://dl.acm.org/citation.cfm?id=644108.644150},
}

@article{ZhangLXKW06,
	author       = {Ying Zhang and
	Xuemin Lin and
	Jian Xu and
	Flip Korn and
	Wei Wang},
	editor       = {Ling Liu and
	Andreas Reuter and
	Kyu{-}Young Whang and
	Jianjun Zhang},
	title        = {Space-efficient Relative Error Order Sketch over Data Streams},
	journal    = {Proceedings of the 22nd International Conference on Data Engineering,
	{ICDE} 2006, 3-8 April 2006, Atlanta, GA, {USA}},
	pages        = {51},
	publisher    = {{IEEE} Computer Society},
	year         = {2006},
	url          = {https://doi.org/10.1109/ICDE.2006.145},
	doi          = {10.1109/ICDE.2006.145},
	bibsource    = {dblp computer science bibliography, https://dblp.org}
}

@article{GK_new,
author = {Gribelyuk, Elena and Sawettamalya, Pachara and Wu, Hongxun and Yu, Huacheng},
title = {Simple \& Optimal Quantile Sketch: Combining Greenwald-Khanna with Khanna-Greenwald},
year = {2024},
issue_date = {May 2024},
publisher = {Association for Computing Machinery},
address = {New York, NY, USA},
volume = {2},
number = {2},
url = {https://doi.org/10.1145/3651610},
doi = {10.1145/3651610},
journal = {Proc. ACM Manag. Data},
month = {5},
articleno = {109},
numpages = {25}
}

@misc{datasketches,
	Author = {Rhodes, Lee and Lang, Kevin and Malkin, Jon and Saydakov, Alexander and Liberty, Edo and Thaler, Justin},
	Howpublished = {Open source software: \url{https://datasketches.apache.org/}},
	Title = {Data{S}ketches: {A} library of stochastic streaming algorithms},
	Year = {2013}}

@article{CormodeKMS05,
	author = {Graham Cormode and
	Flip Korn and
	S. Muthukrishnan and
	Divesh Srivastava},
	editor       = {Karl Aberer and
	Michael J. Franklin and
	Shojiro Nishio},
	title        = {Effective Computation of Biased Quantiles over Data Streams},
	booktitle    = {Proceedings of the 21st International Conference on Data Engineering,
	{ICDE} 2005, 5-8 April 2005, Tokyo, Japan},
	pages        = {20--31},
	publisher    = {{IEEE} Computer Society},
	year         = {2005},
	url          = {https://doi.org/10.1109/ICDE.2005.55},
	doi          = {10.1109/ICDE.2005.55},
	bibsource    = {dblp computer science bibliography, https://dblp.org}
}

@article{HeCCW15,
	author       = {Zaobo He and
	Zhipeng Cai and
	Siyao Cheng and
	Xiaoming Wang},
	title        = {Approximate aggregation for tracking quantiles and range countings
	in wireless sensor networks},
	journal      = {Theor. Comput. Sci.},
	volume       = {607},
	pages        = {381--390},
	year         = {2015},
	url          = {https://doi.org/10.1016/j.tcs.2015.07.056},
	doi          = {10.1016/J.TCS.2015.07.056},
	bibsource    = {dblp computer science bibliography, https://dblp.org}
}

@article{splineSketch,
  author       = {Aleksander Lukasiewicz and
                  Jakub Tetek and
                  Pavel Vesel{\'{y}}},
  title        = {SplineSketch: Even More Accurate Quantiles with Error Guarantees},
  journal      = {Proc. {ACM} Manag. Data},
  volume       = {3},
  number       = {6},
  pages        = {1--26},
  year         = {2025},
  url          = {https://doi.org/10.1145/3769827},
  doi          = {10.1145/3769827},
  bibsource    = {dblp computer science bibliography, https://dblp.org}
}

@article{luo2016quantiles,
  title={Quantiles over data streams: experimental comparisons, new analyses, and further improvements},
  author={Luo, Ge and Wang, Lu and Yi, Ke and Cormode, Graham},
  journal={The VLDB Journal},
  volume={25},
  number={4},
  pages={449--472},
  year={2016},
  publisher={Springer}
}

@misc{bulanek2026matrixfactorizationapproachturnstile,
      title={A Matrix Factorization Approach in Turnstile Streaming}, 
      author={Jan Bulanek and Ravi Kumar and Raghu Meka and Jelani Nelson and Tamas Sarlos},
      year={2026},
      eprint={2607.28819},
      archivePrefix={arXiv},
      primaryClass={cs.DS},
      url={https://arxiv.org/abs/2607.28819}, 
}

\end{document}